\documentclass[10pt,aps,prx,twocolumn,superscriptaddress,nofootinbib,longbibliography]{revtex4-2}

\usepackage{amsmath}
\usepackage{amssymb}
\usepackage{bm}
\usepackage{amsfonts}
\usepackage{graphicx}
\usepackage{xcolor}
\usepackage[normalem]{ulem}
\usepackage{hyperref}
\hypersetup{
    colorlinks=true,
    linkcolor=blue,
    citecolor=blue,
    urlcolor=blue
}

\newcommand{\norm}[1]{\left\lVert#1\right\rVert}

\newcommand{\ket}[1]{|#1\rangle}
\newcommand{\bra}[1]{\langle#1|}

\newcommand{\Tr}{\operatorname{Tr}}

\newtheorem{definition}{Definition}
\newtheorem{theorem}{Theorem}

\newtheorem{lemma}{Lemma}

\begin{document}

\title{Quantum magic and non-commutativity as computational resources in quantum reservoir computing}

\author{Wei Xia}
\thanks{These authors contributed equally to this work}
\affiliation{Department of Physics and The Hong Kong Institute of Quantum Information Science and Technology, The Chinese University of Hong Kong, Shatin, New Territories, Hong Kong, China}

\author{Shuaifan Cao}
\thanks{These authors contributed equally to this work}
\affiliation{State Key Laboratory of Surface Physics, Institute of Nanoelectronics and Quantum Computing, and Department of Physics, Fudan University, Shanghai 200433, China}
\affiliation{Shanghai MicMac Quantum Technology Co., Ltd} 

\author{Xingze Qiu}
\email{xingze@tongji.edu.cn}
\affiliation{School of Physics Science and Engineering, Tongji University, Shanghai 200092, China}

\author{Xiaopeng Li}
\email{xiaopeng\underline{ }li@fudan.edu.cn}
\affiliation{State Key Laboratory of Surface Physics, Institute of Nanoelectronics and Quantum Computing, and Department of Physics, Fudan University, Shanghai 200433, China}
\affiliation{Shanghai Qi Zhi Institute, AI Tower, Xuhui District, Shanghai 200232, China}
\affiliation{Shanghai Buchou Quantum Technology Co., Ltd}
\affiliation{Hefei National Laboratory, Hefei 230088, China}

\date{\today}

\begin{abstract}
Quantum reservoir computing (QRC) provides a hardware-efficient paradigm for temporal information processing on near-term quantum devices. Despite rapid experimental progress, a rigorous understanding of the structural conditions required for its scalable quantum-enhanced performance remains lacking. Here, we develop a theoretical framework in Pauli-Liouville space that provides a unified analytical treatment of the echo state property (ESP), nonlinear expressive power, and quantum resources. 
We first analyze the widely used qubit-resetting scheme and establish that quantum magic generated by reservoir dynamics is a necessary condition for effective computation—a requirement more fundamental than ESP. However, we prove that this architecture faces inherent expressivity limitations: all nonlinear processing originates exclusively from the classical encoding map, imposing an unavoidable trade-off between nonlinearity and memory capacity.
To circumvent this structural bottleneck, we rigorously analyze Hamiltonian encoding, in which temporal inputs are embedded directly into the continuous dynamics generator. 
We show that the ESP is natively guaranteed by the Liouvillian spectral gap, decoupling it from quantum magic. 
Crucially, for any non-trivial drive Hamiltonian, the discrete-time update map exhibits a transcendental, infinite-order nonlinear dependence on the instantaneous input. Moreover, the intrinsic non-commutativity of the open-system generators governs the temporal coupling of these nonlinearities, producing highly non-separable processing of the input history. Our results establish a rigorous theoretical hierarchy of QRC architectures and provide prescriptive design principles for experiments targeting genuine quantum advantages in temporal processing.

\end{abstract}

\maketitle

\section{Introduction}
\label{sec:intro}

\subsection{Quantum Reservoir Computing for Temporal Information Processing}
\label{subsec:qrc_nisq}

Quantum machine learning seeks to exploit the structure of quantum state spaces—high dimensionality, entanglement, and interference—to gain advantages in computational or statistical learning tasks~\cite{biamonte2017quantum,schuld2019quantum,havlicek2019supervised}. Establishing whether and when such advantages are genuine requires precise characterization of the function classes accessible to quantum models relative to their classical counterparts~\cite{servedio2004equivalences,dunjko2016quantum,liu2021rigorous,huang2022quantum}.

Quantum reservoir computing extends this program to temporal information processing, exploiting quantum dynamics as computational resources~\cite{Fujii2017Harnessing}.
In analogy with classical reservoir computing~\cite{Herbert2004Harnessing,Maass2002Real}, QRC processes temporal inputs by sequentially encoding them into a fixed quantum dynamical system, whose evolving observables provide a high-dimensional nonlinear representation of the input history. 
The resulting signals are then combined through a trainable linear readout to perform the target task. 

This architecture is compatible with both analogue quantum devices and gate-based digital quantum circuits, making it an experimentally versatile route toward quantum-enhanced temporal learning. Implementations have been explored in superconducting processors~\cite{Senanian2024Microwave,Dudas2023,Ricci2026Quantum,Suzuki2022Natural}, nuclear magnetic resonance platforms~\cite{Hou2026High,negoro2018machine}, photonic integrated circuits~\cite{Paparelle2026Experimental,DiBartolo2026Time}, and Rydberg atom arrays~\cite{kornjaca2024large}. Its computational capabilities have been demonstrated in tasks including dynamical-system modeling, time-series prediction, and signal processing~\cite{Fujii2017Harnessing,Nakajima2019Boosting,Nokkala2021Gaussian,MartinezPena2021Dynamical,Mujal2023Measurements,Sannia2024Dissipation,Kobayashi2024Feedback}, motivating the search for scalable QRC schemes whose advantages over classical reservoir computing can be rigorously characterized~\cite{Xia2023Configured,Hou2026High}.

\subsection{The Theoretical Gap: Stability, Magic, and Expressivity}
\label{subsec:theoretical_gap}

Despite the rapid progress, the theoretical foundations of QRC remain much underdeveloped. 
Previous theoretical works have largely been confined to employing incomplete mathematical models to explain specific empirical features~\cite{mujal2021analytical, schutte2025expressivity}, or establishing the conditions under which QRC satisfies the echo state property (ESP)~\cite{pena2023quantum, kobayashi2024extending, kobayashi2024coherence,pena2025input}. 
The ESP, which ensures that the reservoir state is asymptotically stable against variations in the initial quantum state, is established as a crucial prerequisite in reservoir computing~\cite{jaeger_2001_esp, Izzet2012revisiting}.
However, the question of which uniquely quantum effects are fundamentally responsible for computational capabilities beyond those of classical reservoir models remains largely open—a gap that must be closed for QRC to provide a scalable, demonstrable advantage.

More specifically, early theoretical investigations attributed the expressive power of QRC to the exponential scaling of the many-body Hilbert space~\cite{Fujii2017Harnessing}, alongside the dynamical generation of quantum coherence~\cite{Palacios2024Role,kora2026quantumness,Xia2023Configured} and multipartite entanglement~\cite{Danilo2025Quantum,askari2025spin,karimi2026role,Kora2024Frequency}. 
However, the Gottesman–Knill theorem shows that even states exhibiting substantial quantum entanglement and coherence can be efficiently simulated classically under Clifford dynamics.~\cite{Nielsen_Chuang_2010}.
This theorem implies that entanglement and coherence, while necessary, are strictly insufficient to guarantee a classically intractable advantage in temporal processing. Consequently, a growing theoretical consensus points toward quantum magic---the rigorous measure of non-stabilizerness---as the ultimate indispensable resource~\cite{Veitch2014resource,Bravyi2005Universal,Aaronson2004Improved}. 
Yet, how magic affects ESP and how it contributes mechanistically to the QRC performance remain unknown. 
Furthermore, since data encoding is central to quantum information processing, identifying which input-encoding paradigm provides a deeper computational advantage is equally pressing. In particular, the fundamental distinction in learning capabilities between the discrete qubit-resetting scheme~\cite{Fujii2017Harnessing} and continuous Hamiltonian modulation~\cite{Bravo2022Quantum,Sannia2024Dissipation} must be understood to establish rigorous design principles for QRC.

\subsection{Summary of Main Results}
\label{subsec:main_results}

In this work, we develop a rigorous mathematical framework operating in the Pauli-Liouville space and analyze the fundamental capabilities of QRC. Our analysis reveals  a sharp distinction between the two standard input-encoding schemes:  qubit resetting and Hamiltonian encoding. 
For qubit resetting, we show QRC suffers from fundamental limitations in nonlinear information processing that can only be partially resolved by introducing quantum magic into the reservoir dynamics.
For Hamiltonian encoding, by contrast, such nonlinear capability emerges naturally from the intrinsic non-commutativity of the physical interactions.
Through both rigorous analysis and numerical verification, we establish quantum magic and non-commutativity as crucial computational resources in QRC, 
thereby delineating the conditions under which QRC can achieve a scalable advantage over classical reservoir computing. Our main results are as follows:

\textit{1. Dynamics-generated magic as the essential prerequisite for QRC in the qubit-resetting architecture.---}
ESP is widely accepted as a vital prerequisite for reliable classical reservoir computing. 
We rigorously show that it is insufficient in the quantum setting: random Clifford dynamics can satisfy the ESP yet suffer catastrophic information loss, rendering the reservoir computationally inert. Instead, it is the dynamics-generated quantum magic—the non-stabilizer resources generated in the quantum reservoir dynamics—rather than  ESP that serves as the true physical prerequisite for sustaining a rich input-history representation and enabling effective computational performance.

\textit{2. The expressivity bottleneck of qubit resetting.---}
We prove a  ``no-go''  theorem for the expressive power of the qubit-resetting architecture. By mapping the quantum dynamics to a classical state-affine system (SAS), we mathematically prove that the nonlinearity of the input-output mapping is strictly bounded by the polynomial degree of the classical encoding function, irrespective of the magic content of the dynamics. The quantum reservoir dynamics do not generate additional nonlinear degrees of freedom; they merely mix existing ones linearly. This inherent limitation dictates that the reservoir explores a strictly constrained function subspace, which implies its performance can, in principle, be matched or surpassed by a purely classical model. Furthermore, it imposes an unavoidable physical trade-off between nonlinearity and temporal memory capacity.

\textit{3. Infinite-order nonlinear expressivity via Hamiltonian encoding.---}
We show Hamiltonian encoding circumvents the expressivity ceiling of qubit resetting.
Here,   quantum magic is generated continuously by the Hamiltonian, while the ESP is enforced by the open-system Liouvillian spectral gap. As a result,  the two properties are structurally decoupled.
Furthermore, by embedding the temporal input directly into the continuous dynamics generator, the discrete-time update map assumes a fundamentally transcendental function form. This generates an infinite-order nonlinear response to the instantaneous drive,  completely evading the finite-degree polynomial ceiling of classical injection maps. The intrinsic non-commutativity of the open-system generators strictly governs the temporal coupling of these nonlinearities across different time steps, synthesizing a complex, non-separable processing of the input history. Consequently, the accessible feature space is no longer restricted to a fixed, finite-dimensional polynomial subspace, empowering the quantum reservoir to support a profound information processing capacity (IPC) hierarchy than is possible under qubit resetting.

\section{Fundamental Limits in Qubit Resetting and the Role of Magic}
\label{sec:magic_necessity}

The Gottesman-Knill theorem dictates that quantum circuits composed exclusively of Clifford gates (generated by $\{\text{CNOT}, H, S\}$) acting on stabilizer states can be simulated efficiently on classical hardware. Consequently, quantum magic (or non-stabilizerness) is universally recognized as a fundamental prerequisite for any potential quantum advantage. 
In this section, we introduce the qubit-resetting QRC architecture alongside the critical concept of dynamics-generated magic. 
To rigorously analyze the QRC system, we first formulate its mathematical framework within the Pauli-Liouville space. 
Building upon this foundation, we analytically demonstrate that in the absence of dynamics-generated magic (i.e., under random Clifford dynamics), a finite fraction of reservoirs can still satisfy the ESP. 
However, these reservoirs suffer from severe information loss, rendering them entirely ineffective as feature generators. 
Therefore, in the quantum regime, dynamics-generated magic constitutes a more fundamental prerequisite for effective reservoir computing than the ESP. 
Furthermore, we prove that even when endowed with sufficient magic, the qubit-resetting architecture still faces an intrinsic expressivity bottleneck stemming fundamentally from the linearity of quantum mechanics. 
Finally, we provide comprehensive numerical tests and benchmarking results that corroborate our theorems and further elucidate the vital role of dynamics-generated magic.

\subsection{The Qubit-Resetting QRC Architecture With/Without Dynamics-Generated Magic}
\label{subsec:qubit_reset_framework} 

Quantum Reservoir Computing adapts the classical reservoir computing paradigm by mapping low-dimensional temporal input sequences, $\{u_k\}_{k=1}^T$ (where $u_k \in [0,1]$), into the exponentially large Hilbert space of a quantum many-body system. In the qubit-resetting protocol, the encoding is non-unitary and local. 

The computation commences with the initialization of the $N$-qubit reservoir into a fiducial state $\rho_0$, typically chosen as the state $\ket{0}^{\otimes N}$ or a thermal state. At each discrete time step $k$, a designated subset of qubits (here, we consider the first qubit without loss of generality) is conceptually ``reset'' or re-initialized to an input-dependent pure state:
\begin{equation}
\ket{\psi^{(1)}_{\rm in}(u_k)} = \sqrt{1 - u_k} \ket{0} + \sqrt{u_k} \ket{1} ,
\label{eq:state_encoding}
\end{equation}
where $\ket{0}$ and $\ket{1}$ denote the standard computational basis states of a single qubit. 
Defining the input state as $\rho^{(1)}_{\rm in}(u_k) = \ket{\psi^{(1)}_{\rm in}(u_k)}\bra{\psi^{(1)}_{\rm in}(u_k)}$, the state injection is formally characterized by a 
local completely-positive-trace-preserving (CPTP) map. 
The state of the input qubit is discarded and replaced by $\rho^{(1)}_{\rm in}(u_k)$, while the remaining $N-1$ qubits preserve their state and correlations. The joint state immediately after the qubit resetting is:
\begin{equation}
\rho'_k = \rho^{(1)}_{\rm in}(u_k) \otimes \Tr_1[\rho_{k-1}] ,
\label{eq:reset_map}
\end{equation}
where $\Tr_1[\cdot]$ denotes the partial trace over the input qubit subspace. Physically, this operation acts as a local, discrete dissipation channel. In the context of open quantum systems, this non-unitary step irreversibly introduces entropy into the local subsystem, erasing old information from the specific input qubit while injecting the fresh input signal; simultaneously, the memory of past inputs is robustly retained in the state of the remaining $N-1$ qubits.

Following the injection, the system evolves under a unitary operator $U$, causing the input information to spread from the input qubit throughout the reservoir.
This unitary evolution effectively mixes the newly injected information with the retained information from past inputs: $\rho_{k} = U \rho'_k U^\dagger$. 
To maximize the accessible feature space without increasing the physical qubit overhead $N$, we employ a temporal multiplexing strategy that exploits the rich transient dynamics of the quantum system~\cite{Fujii2017Harnessing}. The evolution interval is subdivided into $V$ virtual time slots, $U = U_V \cdots U_1$. At each intermediate step $v \in \{1, \dots, V\}$, we measure the local Pauli-$Z$ expectation values $\langle Z_{i,v} \rangle$, concatenating them into a high-dimensional feature vector $\bm{h}_k \in \mathbb{R}^{NV}$.

In reservoir computing, it is required that the extracted feature vector should form an unbiased representation of the input time sequence, i.e., independent of the initial state. A necessary condition is that the reservoir dynamics obey ESP~\cite{jaeger_2001_esp, Izzet2012revisiting}. In classical reservoir computing, this has been taken as a vital prerequisite for effective temporal information processing. For a rigorous analysis of ESP in QRC, we provide a mathematical definition in Section~\ref{subsec:sas_reduction}.

During the training phase, a classical linear readout layer with weight matrix $W_{\text{out}}$ is optimized via ridge regression to map $\bm{h}_k$ to a target sequence $\overline{y}_k$. To rigorously quantify the computational performance of the quantum reservoir, we utilize two complementary metrics. The capacity $C$ measures the linear correlation between the predicted output $y_k$ and the target $\overline{y}_k$, capturing the reservoir's ability to reconstruct the structural features of the target function:
\begin{equation}
C = \frac{\operatorname{cov}^2(y_k, \overline{y}_k)}{\sigma^2(y_k) \cdot \sigma^2(\overline{y}_k)} ,
\label{eq:capacity}
\end{equation}
where $\operatorname{cov}(\cdot)$ and $\sigma^2(\cdot)$ denote the covariance and variance, respectively. A capacity value $C \to 1$ indicates perfect correlation. Additionally, to quantify the precise accuracy of the predictions, we compute the normalized mean square error (NMSE), which normalizes the error against the variance of the target signal:
\begin{equation}
\text{NMSE} = \frac{\sum_k (\overline{y}_k - y_k)^2}{\sum_k \overline{y}_k^2} .
\label{eq:nmse}
\end{equation}
Lower NMSE values correspond to higher predictive accuracy.

Within the proposed architecture, the performance of the quantum reservoir is fundamentally governed by the nature of the unitary operator $U$. 
A paramount characteristic of this operator is its inherent ``quantumness,'' which dictates whether the ensuing dynamics can be efficiently simulated via classical means. 
To rigorously evaluate this, we require a robust metric to quantify the amount of non-Clifford resources present during the evolution.
We adopt the framework of Stabilizer Rényi Entropy (SRE)~\cite{Leone2022Stabilizer}. 
Let $\tilde{\mathcal{P}}_N = \{ \pm 1, \pm i \} \times \{I, X, Y, Z\}^{\otimes N}$ denote the $N$-qubit Pauli group, and let $\mathcal{P}_N := \tilde{\mathcal{P}}_N / \{ \pm 1, \pm i \}$ be the quotient group of Pauli strings, which effectively factors out global phases. 
For a general mixed state $\rho$ in an $N$-qubit system of dimension $d = 2^N$, the 2-Rényi magic is defined as:
\begin{equation}
\tilde{M}_2(\rho) = M_2(\rho) - S_2(\rho) ,
\label{eq:sre_definition}
\end{equation}
where $S_2(\rho) = -\log \Tr[\rho^2]$ is the conventional 2-Rényi purity entropy, and $M_2(\rho)$ is the 2-Rényi SRE, which measures the state's overlap with the stabilizer code subspace:
\begin{equation}
M_2(\rho) = -\log \left[ d \cdot \Tr[Q \rho^{\otimes 4}] \right] .
\label{eq:m2_definition}
\end{equation}
Here, $Q = d^{-2} \sum_{P \in \mathcal{P}_N} P^{\otimes 4}$ denotes the projector onto the stabilizer code subspace. By quantifying the extent to which the quantum state spreads across the Pauli basis beyond the stabilizer structure, $\tilde{M}_2(\rho)$ provides a faithful, Clifford-invariant measure of its magic.

In the qubit-resetting QRC architecture, magic in the reservoir state can originate from three sources: the initial state, the input state $\rho_{\text{in}}(u_k)$, and the reservoir dynamics itself. Of these, 
the dynamical contribution is the most fundamental, as it is intrinsic to the reservoir and independent of the encoding. We refer to this contribution as {\it dynamics-generated magic} and distinguish two scenarios accordingly: Clifford dynamics, in which the reservoir evolution is classically simulable and generates no magic, and universal dynamics, in which the evolution itself produces non-stabilizer resources over time.

\subsection{Operator-Space Formalism and the Role of Magic in QRC Expressivity}
\label{subsec:sas_reduction}

To analyze the exact function form of the reservoir's input-output map, we switch from the density matrix formalism to the Pauli-Liouville operator space. This formulation allows us to precisely track how input information flows and spreads through the operator space of the system. We expand the state of the $N$-qubit reservoir at time $k$, denoted $\rho_k$, in the orthogonal basis of $N$-qubit Pauli strings $\mathcal{P}_N$. The reservoir state vector $r_k \in \mathbb{R}^{4^N}$ is defined via the Hilbert-Schmidt inner product:
\begin{equation}
    r_{k,i} := \frac{1}{\sqrt{2^N}} \Tr\left[ \rho_k B^N_i \right], \quad i \in [4^N]_0,
\label{eq:liouville_state}
\end{equation}
where $\{B^N_i\}$ is a lexicographically ordered basis of the $N$-qubit Pauli strings $\mathcal{P}_N$, and $r_{k,i}$ denotes the $i$-th component of the vector $r_k$. 
Throughout this work, we use the notation $[n]_0:=\{0,1,\ldots,n-1\}$ for the set of the first $n$ nonnegative integers. 
Because any valid quantum state preserves a unit trace, the zeroth component, which corresponds to the global identity operator ($B^N_0 = I^{\otimes N}$), is strictly fixed at $r_{k,0} = 1/\sqrt{2^N}$.

In the standard single-qubit resetting scheme with a single input qubit, the first qubit is periodically re-initialized to an input state $\rho^{(1)}_{\rm in}(u_k) = \frac{1}{2}\left[ \bm{s}_k(u_k) \cdot \bm{\sigma} \right]$, which is the most general form for a single-qubit input state. 
Here, $\bm{\sigma}=(I,X,Y,Z)$ denotes the Pauli-operator vector, while $\bm{s}_k(u_k)=\bigl(1,s_{k,1}(u_k),s_{k,2}(u_k),s_{k,3}(u_k)\bigr)$
is the augmented Bloch vector, with each component $s_{k,\mu}(u_k)\in[-1,1]$ being a scalar function of $u_k$.
For standard amplitude encoding [Eq.~\eqref{eq:state_encoding}], $\bm{s}_k = (1, 2\sqrt{u_k(1-u_k)}, 0, 1-2u_k)$. 
The state update during the input injection phase, $\rho'_k = \rho^{(1)}_{\rm in}(u_k) \otimes \Tr_1[\rho_{k-1}]$ [Eq.~\eqref{eq:reset_map}], acts as a linear transformation $r'_k = \mathcal{R}(u_k) r_{k-1}$ in the operator space. As derived in detail in Appendix~\ref{app:reset_derivation}, the resetting matrix $\mathcal{R}(u_k) \in \mathbb{R}^{4^N \times 4^N}$ assumes a highly sparse block-column structure:
\begin{equation}
    \mathcal{R}(u_k) = \begin{pmatrix}
     I_c & \bm{0} & \bm{0} & \bm{0} \\
     s_{k,1} I_c & \bm{0} & \bm{0} & \bm{0} \\
     s_{k,2} I_c & \bm{0} & \bm{0} & \bm{0} \\
     s_{k,3} I_c & \bm{0} & \bm{0} & \bm{0}
    \end{pmatrix}, \label{eq:reset_matrix}
\end{equation}
where $c := 4^{N-1}$ and $I_c$ is the $c \times c$ identity matrix. This matrix explicitly destroys all information stored in Pauli strings featuring non-identity operators on the first qubit (columns 2, 3, 4). Meanwhile, information residing purely in the $N-1$ un-reset qubits is preserved and broadcast to other sectors (column 1), weighted by the classical input features $\bm{s}_k(u_k)$.

Following the resetting, the system evolves under a unitary $U$, corresponding to an orthogonal transformation $\mathcal{O} \in \mathrm{SO}(4^N)$ with elements $\mathcal{O}_{ij} = 2^{-N} \Tr[ B^N_i U B^N_j U^\dagger ]$ in the operator space. The full state update is $r_k = \mathcal{O} \mathcal{R}(u_k) r_{k-1}$. 
Crucially, because the subsequent resetting step will again definitively discard all but the first $c$ components of $r_k$, we can track only the effective memory state of the $N-1$ un-reset qubits, denoted $x_k \in \mathbb{R}^c$. By partitioning $\mathcal{O}$ into block matrices (see Appendix~\ref{app:sas_derivation}), the quantum reservoir dynamics identically reduce to the following compact, closed-loop recurrence:
\begin{equation}
    \begin{cases}
    x_{k} = (\bm{s}_{k} \cdot \mathbf{A}) \, x_{k-1} + \bm{s}_{k} \cdot \mathbf{b}, \\[6pt]
    y_k = W^\mathrm{T} x_k,
    \end{cases} \label{eq:sas_framework}
\end{equation}
where $\mathbf{A} \in (\mathbb{R}^{c \times c})^4$ and $\mathbf{b} \in (\mathbb{R}^c)^4$ are constant matrices and vectors determined entirely by the unitary $U$, and $W$ is the trained readout weight.

This mathematical reduction is pivotal to our analysis: it establishes an exact isomorphism between the complex quantum qubit-resetting protocol and a classical SAS. Physically, this reveals that the quantum reservoir acts analogously to a classical driven filter whose internal transition rates change linearly with the input Bloch vector $\bm{s}_k(u_k)$. The quantum unitary dynamics, no matter how entangled, do not introduce any intrinsic nonlinear activation; they act purely as linear mixing coefficients.

Within this mathematical framework, let $x_k$ denote the effective state vector at the $k$-th step and we can rigorously describe the architecture of reservoir computing. 
Given a unitary operator $U$, the state-update equation Eq.~\eqref{eq:sas_framework} naturally defines a reservoir map $M_U: [0,1] \times \mathbb{R}^c \to \mathbb{R}^c$ such that $x_k = M_U(u_k, x_{k-1})$. 
In the context of reservoir computing, we begin from an initial state $x_{-\tau}$ at time step $k=-\tau$, where $\tau$ is a positive integer.
By iteratively applying the map $M_U$ driven by the input sequence $\bm{u}=\{u_k\}_{k=-\tau+1}^{0}$, we obtain the full feature vector $x_0=x_0(x_{-\tau},\bm{u}, U)\in\mathbb{R}^c$ at time step $k=0$. Furthermore, this framework provides a basis to rigorously define the Echo State Property (ESP). 

\begin{definition}[Echo State Property]\label{def:esp}
Let $\mathcal{X} \subset \mathbb{R}^c$ be the effective state space of a quantum reservoir with a unitary operator $U$. An output index $i$ of the reservoir satisfies the ESP if, for almost all left-infinite input sequences $\bm{u}=\{u_k\}_{k\in\mathbb{Z}_{\le 0}}$ (with respect to the product Lebesgue measure), the influence of the initial conditions vanishes asymptotically:
\begin{equation}
\lim_{\tau \to \infty} \left| x_{0,i}(x_{-\tau}, \bm{u}, U) - x_{0,i}(x'_{-\tau}, \bm{u}, U) \right| = 0, 
\label{eq:vanishing}
\end{equation}
for any pair of initial states $x_{-\tau}, x'_{-\tau} \in \mathcal{X}$. Here, $x_{0,i}(x_{-\tau}, \bm{u})$ denotes the $i$-th component of the state vector at time $k=0$, resulting from evolving the initial state $x_{-\tau}$ from time $k=-\tau$ under $M_U$ driven by the input sequence $\bm{u}$. A quantum reservoir is said to satisfy the ESP if all of its output indices satisfy this property.
\end{definition}

Slightly different from the conventional definitions of the ESP in previous works~\cite{jaeger_2001_esp, chen2019learning, pena2023quantum}, we require the vanishing condition Eq.~\eqref{eq:vanishing} to hold for \textit{almost all} input sequences. 
This measure-theoretic relaxation rigorously excludes pathological, measure-zero subsets of inputs (e.g., strictly constant sequences) that could artificially trap the dynamics and prevent convergence. Physically, this condition is well-justified, as the temporal signals processed by a practical quantum reservoir inherently exhibit fluctuations.  
Furthermore, while Definition~\ref{def:esp} evaluates the asymptotic convergence at the specific time step $k=0$, this choice is made strictly without loss of generality. Because the underlying reservoir map is time-translation invariant, the asymptotic convergence and all subsequent analytical results derived herein apply equally to the reservoir state $x_k$ at any arbitrary finite time step $k \in \mathbb{Z}$.


Now, we can analytically elucidate why the qubit-resetting QRC without dynamics-generated magic (i.e., with random Clifford circuits) can satisfy the ESP, yet fundamentally fail to yield an effective quantum reservoir. A random Clifford circuit $U \in \mathcal{C}_N$ normalizes the Pauli group, inducing a strict bijection $f: [4^N]_0 \to [4^N]_0$ on the Pauli indices along with a phase function $\eta: [4^N]_0 \to \{-1,+1\}$, such that $U^\dagger B^N_i U = \eta(i) B^N_{f(i)}$. Consequently, the action of a Clifford circuit maps the input history backward in time through discrete, non-branching permutation trajectories in operator space.

\begin{theorem}[QRC under Clifford Dynamics] \label{thm:esp_without_magic}
Consider an $N$-qubit quantum reservoir governed by a random Clifford circuit $U\in \mathcal{C}_N$, employing a single-qubit resetting scheme with the input state $\rho^{(1)}_{\rm in}(u_k) = \frac{1}{2}(I + u_k Z)$. In the thermodynamic limit ($N \to \infty$), the reservoir satisfies the ESP with a probability of $3/4$. 
Conditioned on the ESP being satisfied, for any left-infinite input sequences $\bm{u}=\{u_k\}_{k\in\mathbb{Z}_{\le 0}}$ and any initial state $x_{-\tau}$, the expected number of nonzero Pauli measurements at time step $k=0$, averaged over the Clifford circuit ensemble, is bounded above by 3 in the thermodynamic limit. Formally, 
\begin{equation}
    \lim_{N \to \infty} \mathbb{E}_{U\in \mathcal{C}_N} \left[ \lim_{\tau \to \infty} \Vert{} x_0(x_{-\tau}, \bm{u}, U)\Vert{}_0 \right] \le 3,
\end{equation}
where $ \| a\|_0$ denotes the number of nonzero entries of $a$. 
\end{theorem}

\textit{Proof.} 
A rigorous combinatorial proof tracing back the trajectory of any output index $i$ is provided in Appendix~\ref{app:proof_thm_esp}. \hfill $\blacksquare$

For concreteness, we have considered the simplest $Z$-basis encoding. 
For a different single-qubit input state, the probability $3/4$ and the expected number $3$ may shift to a different value of order $O(1)$ with respect to the system size $N$. Details are provided in Appendix~\ref{app:proof_thm_esp}. 

Theorem~\ref{thm:esp_without_magic} demonstrates that part of the Clifford reservoirs can indeed satisfy the ESP, a direct consequence of their dynamics strictly permuting Pauli strings rather than coherently scrambling them. 
However, it also proves that the ESP alone is fundamentally insufficient for effective quantum temporal processing. In the absence of non-Clifford resources, the reservoir suffers from catastrophic information loss, yielding that the expected number of nonzero Pauli components remains strictly $O(1)$ regardless of the system size $N$. Crucially, the nonzero components of a random Clifford reservoir are randomly scattered across the $4^N - 1$ non-identity Pauli bases. While general QRC standardly utilizes local observables as measurement features, their expectation values in this case are typically zero due to the exponentially large operator space. Therefore, a random Clifford circuit fails to serve as an effective reservoir.

Furthermore, the bijective nature of the Clifford circuit, along with the sparse block-column structure of the resetting operation [Eq.~\eqref{eq:reset_matrix}], allows a classical computer to efficiently trace back the trajectory of any output index $i$. 
Consequently, a Clifford reservoir computer satisfying the ESP can be efficiently simulated classically.

These limitations cannot be circumvented either by the magic in the initial state of the reservoir, which will be erased by the ESP, or by the magic injected through $\rho^{(1)}_{\rm in}(u_k)$, which simply scales the Pauli coefficients according to Eq.~\eqref{eq:reset_matrix} and cannot alter the underlying permutation structure $f$ of the Clifford unitary. 
These two sources of magic are fundamentally incapable of producing an effective reservoir or overcoming classical simulability.
Genuine quantum advantage requires magic to be generated \textit{intrinsically} by non-Clifford reservoir dynamics, which provides the only mechanism capable of expanding the accessible Liouville space and autonomously generating complex many-body correlations.

The preceding analysis establishes dynamics-generated magic as a more fundamental prerequisite for reservoir computing than the ESP in the quantum regime. 
However, even with sufficient dynamics-generated magic, the qubit-resetting architecture still has fundamental constraints. 
We rigorously prove this structural expressivity limit in the following. 

The origin of this limitation is not the absence of quantum resources, but the specific way in which the classical input is injected: the reset map introduces input dependence only through a finite-dimensional classical encoding Bloch vector $\bm{s}_k$, while the subsequent quantum evolution acts linearly in Pauli-Liouville space. Consequently, the accessible input-output functions remain confined to a bounded-degree polynomial class. To make this limitation explicit, we generalize the single-qubit resetting scheme to $m$ input qubits, with $m<N$, which are discarded and replaced by input-dependent states at each discrete time step. 

\begin{theorem}[Expressivity Limit of Qubit Resetting]\label{thm:expressivity_limit}
Consider an $N$-qubit reservoir with $m$ input-qubits, which are reset at each step according to the input time sequence. The qubit resetting is assumed to take a standard form, 
$\rho^{(m)}_{\rm in}(u_k) = \left[ \frac{1}{2}(I + u_k Z) \right]^{\otimes m}$. 
Then, the expectation value of any reservoir observable at time $k=0$ belongs strictly to the polynomial space:
\begin{equation}
    x_{0,i} \in \mathbb{R}_m\!\left[u_{-\infty},\ldots,u_0\right], \quad \forall\, i,
\end{equation}
where $\mathbb{R}_m[u_{-\infty},\ldots,u_0]$ denotes the space of real polynomials depending on finitely many past inputs, such that the maximum degree of each individual variable $u_k$ is strictly bounded by $m$.
\end{theorem}

\textit{Proof.} Iterating the $m$-qubits extension of the SAS recurrence backward in time generates a Volterra-like series. Because the state injection acts as a linear broadcast in operator space at each time step, variables associated with the exact same instantaneous input $u_k$ are never multiplied together during the dynamical evolution. Summing over the system history thus preserves a strict polynomial degree bound. The complete proof is detailed in Appendix~\ref{app:proof_thm_expressivity}. \hfill $\blacksquare$

To rigorously isolate the intrinsic nonlinearity of the reservoir dynamics,
our theorem introduces a minimal encoding scheme for the $m$ input qubits, ensuring that no extraneous nonlinearity originates from the data encoding process. 
In the fundamental single-qubit case ($m=1$), we adopt an encoding where the Bloch vector $\bm{s}_k(u_k)=(1,0,0,u_k)$ scales strictly linearly with the input $u_k$. 
The multi-qubit input state ($m>1$), denoted as $\rho^{(m)}_{\rm in}(u_k)$, is then constructed via the tensor product of these single-qubit states. 
Without loss of generality, this theorem naturally generalizes to an arbitrary input state.
Specifically, for an arbitrary $m$-qubit input state, the polynomial space generalizes from the original $\mathbb{R}_m[\{u_k\}_{k\in\mathbb{Z}_{\le 0}}]$ to $\mathbb{R}_m[\{s_{k,\mu}\}_{k\in\mathbb{Z}_{\le 0}, \mu}]$.

Theorem~\ref{thm:expressivity_limit} establishes a central structural limitation of the qubit-resetting architecture. It reveals that any nonlinear dependence of the output on the temporal input sequence is inherited entirely from the classical encoding map; the quantum reservoir dynamics simply determine how these pre-existing nonlinear features are linearly mixed. Consequently, the input-output functions generated by this architecture are rigorously confined to the bounded-degree polynomial space $\mathbb{R}_m$. The hypothesis class of the quantum reservoir is thus restricted to a mere subspace of $\mathbb{R}_m$.

This structural limitation imposes severe consequences for classical simulability at the level of function classes. For temporal tasks requiring a fixed memory window of length $L$, the space $\mathbb{R}_m$ admits a classically constructible basis consisting of $(m+1)^L$ monomials, which scales polynomially with the number of reset qubits $m$ and does not scale with the system size $N$. Therefore, one can always construct a purely classical model whose hypothesis class fully encapsulates that of the qubit-resetting reservoir by simply computing these polynomial features and applying a classical linear readout. After training over the same memory window, such a classical model is mathematically guaranteed to match or even outperform the best predictor obtainable from the quantum reservoir, as the QRC system only explores a constrained subset of the exact same ambient function space. 
In fact, a model sharing the same conceptual core has already been introduced in the context of classical reservoir computing~\cite{gauthier2021next}. Qubit-resetting QRC therefore offers no genuine quantum advantage at the level of accessible input–output function classes.

To enhance the nonlinearity (i.e., the accessible polynomial degree $m$) of a qubit-resetting QRC, one is forced to simultaneously inject the input into a larger subset of qubits. However, physically resetting $m$ qubits actively destroys the quantum information residing in that fraction of the system. This inherently increases the local dissipation rate, severely degrading the reservoir's temporal memory capacity. This leads to an unavoidable trade-off: \textit{Within the qubit-resetting paradigm, a quantum reservoir cannot simultaneously achieve high nonlinear expressivity and a long temporal memory capacity}.

\begin{figure*}[t]
\centering
\includegraphics[width=1\textwidth]{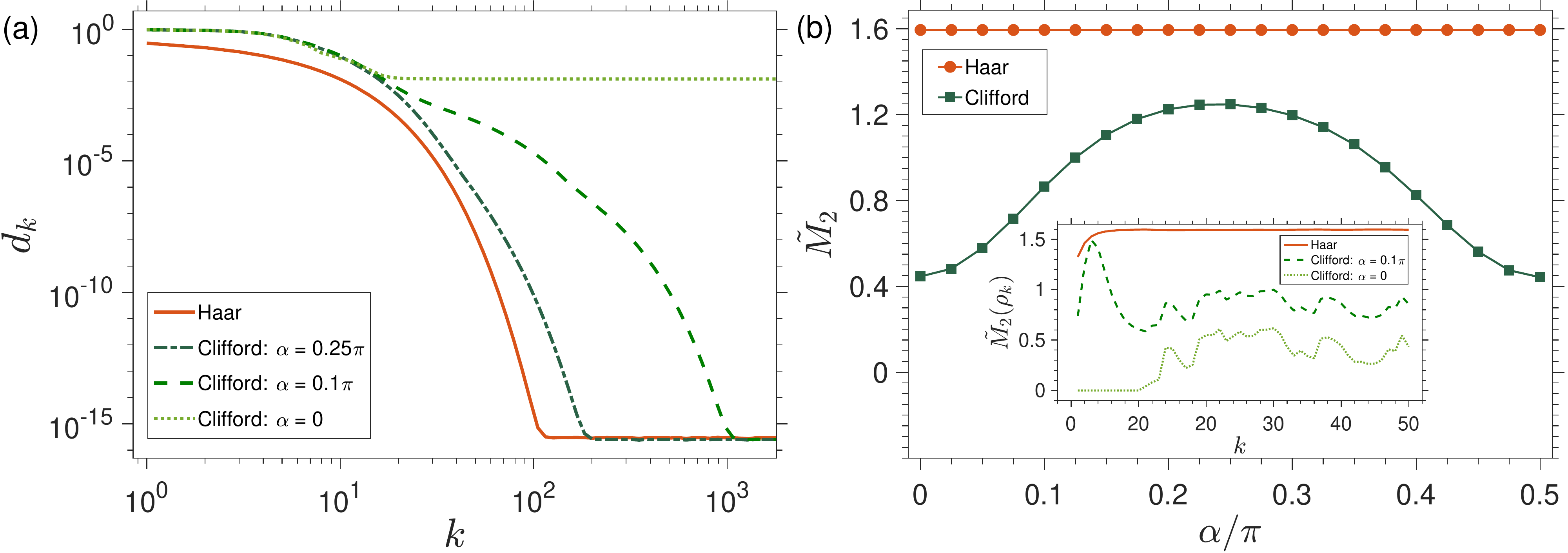}
\caption{
\textbf{Relationship between dynamics-generated quantum magic and the echo state property (ESP).} 
(a) Temporal evolution of the trace distance $d_k$ between two reservoir states prepared in distinct conditions. The reservoir dynamics correspond to a parameterized unitary family $U(\alpha)=R_z(\alpha)C$, which is compared against a Haar-random benchmark.
In the pure Clifford limit ($\alpha=0$), the ensemble-averaged $d_k$ rapidly converges to a finite, nonzero asymptotic value ($\sim10^{-2}$), reflecting the contribution of Clifford realizations that fail to erase the initial condition; this averaged behavior is consistent with Theorem~\ref{thm:esp_without_magic}, where only a finite fraction of Clifford reservoirs satisfy the ESP. In the non-Clifford case ($\alpha\neq0$), $d_k$ decays to a negligible value $(\sim10^{-15})$, providing numerical evidence for the ESP.  
The Haar-random benchmark exhibits the most rapid decay. 
(b) Time-averaged 2-R\'enyi magic $\tilde{M}_2$ as a function of the tuning angle $\alpha$. The temporal average is computed as $K^{-1}\sum_{k=1}^{K}\tilde{M}_2(\rho_k)$ over $K=10^4$ simulated time steps. 
The magic is minimal in the non-Clifford limit ($\alpha =0 $ and $0.5\pi$), which is generated purely from data injection. It becomes larger as we move away from the non-Clifford limit. 
Inset: Time evolution of the instantaneous magic $\tilde{M}_2(\rho_k)$ for representative Clifford and Haar-random reservoirs. Across all panels, data points are ensemble-averaged over 200 independent joint samples with system size $N = 6$. Each sample consists of a random unitary realization paired with an independently generated input sequence, where the inputs $u_k$ are drawn uniformly from $[0,1]$.}
\label{Fig_1}
\end{figure*}

\subsection{Numerical Tests and Benchmark Performance}
\label{subsec:magic_and_esp}

We now numerically examine the role of dynamics-generated magic in the qubit-resetting architecture. The goal of this subsection is twofold: first, to test how non-Clifford dynamics affects ESP, and second, to evaluate whether the same magic generation is reflected in practical benchmark performance. The numerical validation of the intrinsic expressivity bottleneck established in Theorem~\ref{thm:expressivity_limit} is deferred to Sec.~\ref{subsec:information_processing_capacity}, where the qubit-resetting and Hamiltonian-encoding architectures are directly compared within the same IPC framework. Additional numerical checks of Theorem~\ref{thm:esp_without_magic}, including the finite-size convergence of $p_{\rm ESP}$ to $3/4$ and the $O(1)$ scaling of the surviving nonzero Pauli components, are provided in Appendix~\ref{app:proof_thm_esp}.

To diagnose the ESP operationally, we monitor the trace distance
$
d_k=\frac{1}{2}|\rho_k-\tilde{\rho}_k|_1
$
between two reservoir trajectories initialized from different states, $\rho_0$ and $\tilde{\rho}_0$, but driven by the same input sequence. We consider the tunable unitary family
$
U(\alpha)=R_z(\alpha)C,
$
where $C$ is a random Clifford circuit and $R_z(\alpha)=\exp[-i(\alpha/2)\sum_j Z_j]$ is a global non-Clifford rotation. The angle $\alpha$ therefore controls the amount of dynamics-generated magic. As a fully chaotic benchmark, we also compare with Haar-random unitaries.

Figure~\ref{Fig_1}(a) shows a clear separation between the Clifford and non-Clifford regimes. In the Clifford limit $\alpha=0$, the ensemble-averaged trace distance approaches a finite nonzero value, reflecting the finite fraction of Clifford reservoirs that fail to erase the initial condition, consistent with Theorem~\ref{thm:esp_without_magic}. Once non-Clifford rotations are introduced, the trace distance decays to zero for the sampled dynamics, indicating restoration of the ESP. Although the generated magic depends non-monotonically on $\alpha$, the decay rate increases with the amount of dynamics-generated magic along this parameter sweep: larger $\tilde{M}_2$ corresponds to faster decay of 
the trace distance, $d_k$. The Haar-random benchmark exhibits the fastest convergence.

Figure~\ref{Fig_1}(b) further clarifies the distinction between the dynamics-generated magic and the input-generated magic, which is brought into the system by the input state. Even when $\alpha=0$, the input injection can produce a finite baseline of state magic, because generic input states are non-stabilizer states. However, Clifford dynamics cannot amplify this non-stabilizerness localized on the input qubit into many-body operator spreading. Consequently, this input-generated magic alone 
cannot remedy the information-loss issue established in Theorem~\ref{thm:esp_without_magic}. 
Non-Clifford dynamics, by contrast, delocalize Pauli strings across operator space, enabling the reset-induced dissipation to propagate beyond the local injection site and thereby promote the ESP.

\begin{figure*}
\centering
\includegraphics[width=1\textwidth]{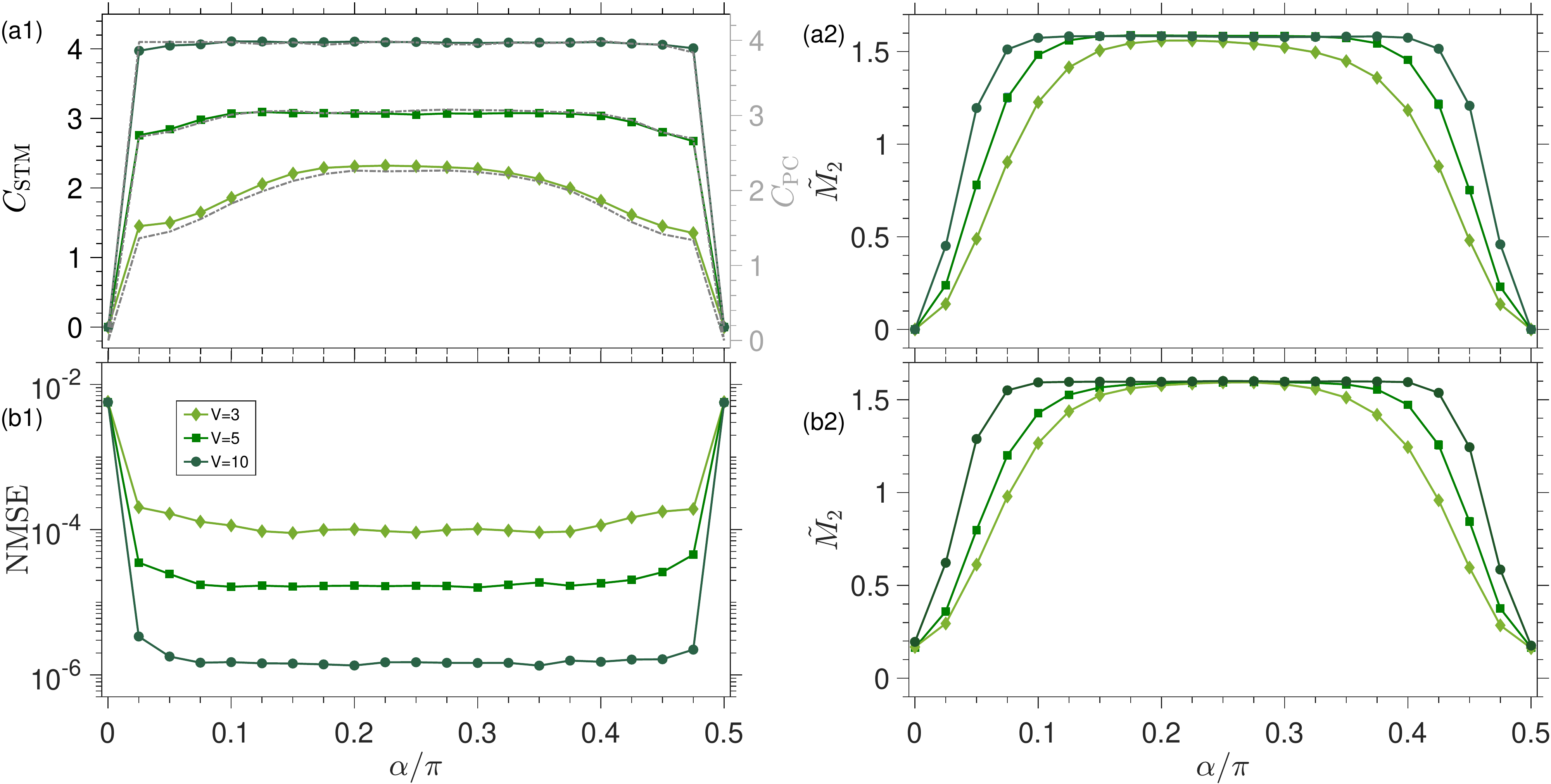}
\caption{
\textbf{Benchmark performance and dynamics-generated quantum magic under qubit resetting.} 
(a1) Short-term memory (STM) capacity $C_{\mathrm{STM}}$ (left axis) and total parity check (PC) capacity $C_{\mathrm{PC}}$ (right axis) as functions of the tuning angle $\alpha$ for reservoirs driven by the parameterized unitary family $U(\alpha)$. 
(a2)  
The corresponding 2-R\'enyi magic $\tilde{M}_2$ (time-averaged) generated during the STM and PC benchmarks.  
(b1) Normalized mean square error (NMSE) for the NARMA task as a function of $\alpha$. 
(b2) The 2-R\'enyi magic $\tilde{M}_2$ generated during the NARMA benchmark. 
For both the STM and PC tasks, the discrete input signal $u_k$ is drawn from an independent and identically distributed (i.i.d.) uniform binary sequence. For the NARMA task, the input drive is defined as $u_k = 0.2\left[\sin\!\left(\frac{2\pi \alpha_2 k}{T}\right)\sin\!\left(\frac{2\pi \beta_2 k}{T}\right)\sin\!\left(\frac{2\pi \gamma_2 k}{T}\right)+1\right]$, with temporal parameters $(\alpha_2,\beta_2,\gamma_2,T)=(2.11,3.73,4.11,100)$. The corresponding NARMA target-function hyperparameters are set to $(\alpha_1,\beta_1,\gamma_1,\delta,n)=(0.3,0.05,1.5,0.1,10)$. All simulations employ a system size of $N=6$ qubits and are ensemble-averaged over 200 independent random circuit realizations. Across all three computational benchmarks, the washout, training, and testing phases consist of 1000, 3000, and 1000 discrete time steps, respectively. Distinct curves across the panels denote different temporal multiplexing densities, parameterized by the number of virtual time slots $V\in\{3,5,10\}$.
}
\label{Fig_2}
\end{figure*}

We next investigate whether dynamics-generated magic is also reflected in task-level performance. We evaluate the qubit-resetting reservoir on three standard temporal benchmarks using the same tunable non-Clifford dynamics, where the evolution within one input interval is implemented as $U(\alpha)=R_z(\alpha)C_V\cdots R_z(\alpha)C_1$. Here, $V$ is the number of virtual time slots. The benchmarks include short-term memory (STM), which probes linear memory; parity check (PC), which probes nonlinear temporal processing; and NARMA, which combines memory with nonlinear cross-time dependence. The results are summarized in Fig.~\ref{Fig_2}. Across all three tasks, moving away from the Clifford point enhances performance: the STM and PC capacities increase, while the NARMA prediction error decreases. These improvements correlate with the build-up of dynamics-generated magic, supporting the interpretation that non-Clifford operator spreading is a key resource for effective feature generation in qubit-resetting QRC.

The dependence on temporal multiplexing is also consistent with this picture. Increasing the number of virtual nodes $V$ improves the benchmark scores by exposing a larger set of intermediate dynamical observables to the linear readout. Although the time-averaged magic can saturate for intermediate values of $\alpha$, the performance may continue to improve with $V$, because $V$ controls not only the amount of generated non-stabilizerness but also how efficiently the magic-enabled dynamics are sampled and extracted by the readout layer.

These numerical results corroborate the theoretical analysis above. In the qubit-resetting architecture, dynamics-generated magic promotes operator scrambling and distributes the reset-induced dissipation beyond the injection site, sustaining nontrivial history-dependent representations. The structural expressivity bound of Theorem~\ref{thm:expressivity_limit} nevertheless persists: the accessible nonlinear function class is constrained by the classical encoding map, not by the reservoir dynamics. This limitation motivates the Hamiltonian-encoding architecture studied next, where the input enters through the generator of the continuous-time evolution rather than through the destructive local qubit-resetting.

\section{Breaking the Bottleneck: QRC with Hamiltonian Encoding}
\label{sec:hamiltonian_encoding}

The rigorous operator-space analysis in Section~\ref{sec:magic_necessity} revealed a fundamental structural limitation of the qubit-resetting architecture: its nonlinear expressivity is strictly bounded by the number of reset input qubits ($m$), imposing an unavoidable structural trade-off between nonlinearity and temporal memory capacity. 
Furthermore, as established in Section~\ref{sec:magic_necessity}, effective qubit-resetting reservoirs rely on dynamics-generated magic to promote operator scrambling, distribute the local reset-induced dissipation, and support rich history-dependent state representations.

To achieve genuine, scalable quantum advantage in temporal processing, a QRC architecture must fundamentally bypass these mathematical bottlenecks. 
In this section, we introduce and rigorously analyze the Hamiltonian encoding architecture. Instead of destructive, local state replacements, this paradigm embeds the discretized temporal input sequence as piecewise-constant continuous drives directly into the generator of the open-system dynamics. We demonstrate that this approach elegantly decouples the ESP from the magic, as the ESP is natively enforced by the Liouvillian spectral gap and the magic arises from the continuous-time dynamics. 
Furthermore, we rigorously prove that provided the drive Hamiltonian is non-trivial (i.e., not proportional to the identity), this continuous embedding yields a transcendental function form that autonomously generates an infinite-order nonlinear response to the instantaneous input. Crucially, the intrinsic non-commutativity of the open-system generators strictly governs the temporal coupling of these nonlinearities, synthesizing a complex, non-separable processing of the input history that fundamentally overcomes the finite-degree expressivity limits of the qubit-resetting architecture.

\subsection{The Hamiltonian Encoding QRC Architecture}
\label{subsec:hamiltonian_framework}

In the Hamiltonian encoding scheme, the reservoir is modeled as an open quantum many-body system continuously interacting with a Markovian environment. The sequential input signal $u_k$ is encoded directly into the amplitudes of the global or local driving fields. Over the duration of the $k$-th time step, the continuous-time evolution of the system's density matrix $\rho(t)$ is naturally governed by the Lindblad master equation:
\begin{align}
    \frac{d\rho(t)}{dt} &= \mathcal{L}\big(u_k\big) \rho(t) \nonumber \\
    &= -i \big[H_0 + u_k H_1, \, \rho(t)\big] + \sum_\mu \gamma_\mu \mathcal{D}[L_\mu] \rho(t),
\label{eq:lindblad_master}
\end{align}
where $H_0$ is the unperturbed native Hamiltonian of the many-body system (capturing internal interactions), and $H_1$ is the drive Hamiltonian modulated by the classical temporal input $u_k$. The environmental coupling is described by the standard dissipators $\mathcal{D}[L_\mu]\rho = L_\mu \rho L_\mu^\dagger - \frac{1}{2}\{L_\mu^\dagger L_\mu, \rho\}$, with physical jump operators $L_\mu$ (e.g., spontaneous emission or dephasing) and decay rates $\gamma_\mu > 0$.

In the Pauli-Liouville representation, the density matrix $\rho(t)$ maps to a real vector $r(t) \in \mathbb{R}^{4^N}$, and the Lindbladian superoperator becomes a real matrix $\mathbf{L}(u_k) = \mathbf{L}_0 + u_k \mathbf{L}_1$, 
where $\mathbf{L}_0$ is the input-independent part of the generator and $\mathbf{L}_1$ couples the input linearly to the dynamics. As detailed in Appendix~\ref{app:hamiltonian_encoding_proof}, the evolution equation assumes a strictly affine generator form:
\begin{equation}
    \frac{dr(t)}{dt} = \Big( \mathbf{L}_0 + u_k \mathbf{L}_1 \Big) r(t).
\label{eq:liouville_diff_eq}
\end{equation}
In standard digital QRC protocols, the continuous input is held piecewise constant over a discrete injection interval $\Delta t$. For the $k$-th time step, the discrete update map emerges as the exact matrix exponential of the generator:
\begin{equation}
    r_k = \mathcal{E}(u_k) r_{k-1} = \exp\Big[ \Delta t (\mathbf{L}_0 + u_k \mathbf{L}_1) \Big] r_{k-1}.
\label{eq:hamiltonian_discrete_map}
\end{equation}
As demonstrated below, this formulation ensures that the Hamiltonian encoding architecture avoids the limitations inherent in the qubit-resetting architecture.

A primary theoretical advantage of the Hamiltonian encoding scheme lies in the purely physical origin of its ESP. In the qubit-resetting scheme, dissipation via the partial trace is entirely local to the input nodes. Consequently, the system must rely heavily on non-Clifford unitary scrambling (magic) to distribute this local dissipation across the operator space and erase initial-state information.

In stark contrast, the quantum reservoir dynamics in the Hamiltonian encoding scheme are described by a Markovian open quantum system governed by the instantaneous Lindbladian $\mathbf{L}(u_k) = \mathbf{L}_0 + u_k \mathbf{L}_1$. Physically, the unperturbed generator $\mathbf{L}_0$ characterizes the intrinsic many-body evolution and the relaxation channels coupling the system to its environment. Assuming these environmental interactions are sufficient to drive the unperturbed system to a unique steady state, $\mathbf{L}_0$ possesses a strictly positive Liouvillian spectral gap $\lambda_0 > 0$, defined by the magnitude of the largest nonzero real part among its non-trivial eigenvalues.

The temporal input sequence $\{u_k\}$ introduces a continuous deformation to the dynamics via the drive generator $u_k \mathbf{L}_1$. Provided that the drive amplitude is suitably bounded relative to $\lambda_0$, the instantaneous generator $\mathbf{L}(u_k)$ preserves the uniqueness of its steady state and maintains a uniform, strictly positive dissipative gap $\lambda_{\min} > 0$ across all possible input values. 
Under these rigorously defined physical conditions, the discrete-time update operator $\mathcal{E}(u_k) = \exp[\Delta t (\mathbf{L}_0 + u_k \mathbf{L}_1)]$ constitutes a strictly contractive CPTP map. For any two physical trajectories initialized from distinct states $\rho_0$ and $\tilde{\rho}_0$ but driven by the exact same input sequence, the trace distance $d_k = \frac{1}{2} \|\rho_k - \tilde{\rho}_k\|_1$ inherently satisfies $d_k \le \kappa \, d_{k-1}$, where the uniform contraction factor $\kappa < 1$ is bounded strictly away from unity due to the persistent spectral gap~\cite{Sannia2024Dissipation,chen2022nonlinear}. Iterating this relation yields an exponential suppression of the initial-state memory, mathematically guaranteeing that $\lim_{k \to \infty} d_k = 0$.

Therefore, in the Hamiltonian encoding paradigm, the ESP is strictly enforced by the native dissipative contraction of the open-system Liouvillian. This mechanism fundamentally decouples the emergence of the ESP from the requirement of dynamics-generated quantum magic, bypassing the intrinsic memory-nonlinearity trade-off that limits the qubit-resetting architecture.

\subsection{Expressivity of Hamiltonian Encoding}
\label{subsec:infinite_expressivity}

With the ESP guaranteed by the Liouvillian spectral gap, we now turn to expressive power. Because the input enters through the generator of the continuous-time evolution—whose terms need not commute—the resulting discrete-time map acquires a transcendental dependence on the instantaneous drive, producing nonlinear responses at all polynomial orders and circumventing the expressivity bound of Theorem~\ref{thm:expressivity_limit}. The following theorem makes this precise.

\begin{theorem}[Expressivity of Hamiltonian Encoding]
\label{thm:hamiltonian_pauli_expansion_main}
Consider a Hamiltonian-encoded reservoir governed by Eq.~\eqref{eq:lindblad_master}. Let $r_k \in \mathbb{R}^{4^N}$ denote the Pauli-Liouville state of the reservoir at discrete time step $k$, driven by the instantaneous scalar input $u_k \in \mathbb{R}$. Assume the drive Hamiltonian is non-trivial, 
i.e., $H_1 \neq c I$ for any real constant $c$.
Then, the one-step update map can be strictly expanded as:
\begin{equation}
r_{k} = \left( \sum_{m=0}^{\infty} C_m u_k^m \right) r_{k-1},
\label{eq:ham_map}
\end{equation}
where the coefficient matrices $C_m$ are uniquely determined by the reservoir's internal dynamics. Crucially, this power series does not 
terminate at any finite order $m$. The discrete-time update map thus constitutes a fundamentally transcendental matrix function, inherently equipping the system with an infinite-order nonlinear response to the instantaneous input $u_k$.
\end{theorem} 

\textit{Proof.} 
The rigorous proof, detailed in Appendix~\ref{app:hamiltonian_encoding_proof}, proceeds by contradiction. First, by mapping the Lindblad dynamics into the Pauli-Liouville space, we demonstrate that the discrete one-step update is governed by a matrix exponential whose generator depends strictly linearly on the input $u_k$ [Eq.~\eqref{eq:hamiltonian_discrete_map}]. The non-triviality condition ($H_1 \neq c I$) guarantees that the input generator $\mathbf{L}_1$ possesses at least one nonzero eigenvalue. 
Then, complexifying $u_k$ to $z\in\mathbb{C}$ shows that the analytic continuation of the matrix exponential $e^{\Delta t(\mathbf{L}_0+z\mathbf{L}_1)}$ has exponential spectral-norm growth along a suitable complex ray. In contrast, if the update map were an algebraic matrix function, Puiseux's theorem~\cite{brieskorn1986plane} dictates that it could be expanded near infinity as a convergent fractional power series.
Consequently, along the same ray, its norm would grow at most as a finite power of $u_k$. 
This profound discrepancy in asymptotic growth yields a mathematical contradiction, establishing the update map as a transcendental function of the input and rigorously precluding its power series expansion from terminating at any finite order. \hfill $\blacksquare$

This theorem formally identifies the essential structural advantage of Hamiltonian encoding: by embedding the input directly into the dynamical generator, the reservoir completely evades the finite-order nonlinear ceiling. Consequently, the accessible input-output mapping is not restricted \textit{a priori} to any predetermined finite-dimensional polynomial family. In this sense, Hamiltonian encoding natively supports a qualitatively richer hierarchy of complex temporal features.

This stands in stark contrast to the qubit-resetting architecture. In that paradigm, the quantum dynamics merely execute a linear mixing of the nonlinear features already injected by the classical encoding step, strictly confining the accessible hypothesis class to a finite-dimensional polynomial space. Therefore, one can always construct a purely classical model whose function class fully encapsulates that of the qubit-resetting reservoir. No analogous finite-dimensional function-space bound exists for Hamiltonian encoding, whose nonlinear hierarchy is not limited by any fixed polynomial degree.
Although this unbounded expressivity does not autonomously guarantee classical intractability for every specific computational instance, it rigorously proves that the specific structural bottleneck crippling the qubit-resetting scheme is fundamentally absent. Hamiltonian encoding thus provides a principled, mathematically sound pathway toward achieving genuine function-class advantage by accessing a substantially larger and more complex family of temporal mappings. 

While Theorem~\ref{thm:hamiltonian_pauli_expansion_main} rigorously establishes an infinite-order nonlinear response with respect to the instantaneous input $u_k$, it is crucial to distinguish this single-step property from the reservoir's broader capacity to process temporal sequences. The unbounded instantaneous expressivity stems solely from the non-triviality of the drive ($\mathbf{L}_1 \neq \mathbf{0}$). However, the genuine computational power of a temporal filter relies on its ability to generate complex, non-separable cross-terms between inputs injected at different time steps (e.g., $u_k^m u_{k-1}^n$). This cross-temporal coupling is governed entirely by the intrinsic non-commutativity of the open-system generators, $[\mathbf{L}_0, \mathbf{L}_1] \neq 0$. 
To rigorously illustrate this, consider the counterfactual scenario where the intrinsic dynamics and the drive commute, $[\mathbf{L}_0, \mathbf{L}_1] = 0$. The discrete-time evolution over multiple steps would analytically factorize:
\begin{equation}
    r_k = \exp\left[k\Delta t \mathbf{L}_0 \right] \exp\left[ \Delta t \left( \sum_{j=1}^k u_j \right) \mathbf{L}_1 \right] r_0.
\end{equation}
In this commuting limit, the reservoir's memory reduces to a trivial linear accumulation of past inputs. Although the mathematical dependence on the aggregate sum $\sum u_j$ remains infinite-order, the system becomes structurally incapable of generating distinct multiplicative cross-terms between different time steps. This would severely cripple its expressivity, rendering it unable to differentiate complex temporal patterns.

In reality, the physical generators of the driven many-body system strictly do not commute ($[\mathbf{L}_0, \mathbf{L}_1] \neq 0$). Consequently, the multi-step propagator cannot be simply factorized. Expanding the time-ordered product $\prod_{j} \exp[\Delta t(\mathbf{L}_0 + u_j \mathbf{L}_1)]$---via the Baker-Campbell-Hausdorff formula or the Dyson series---generates an infinite hierarchy of nested commutators, such as $[\mathbf{L}_0, \mathbf{L}_1]$ and $[\mathbf{L}_1, [\mathbf{L}_0, \mathbf{L}_1]]$. Physically, these non-vanishing commutators act as an essential dynamical mixing mechanism. They ensure that the transcendental features generated at previous time steps are highly scrambled by the intrinsic Hamiltonian and dissipative evolution ($\mathbf{L}_0$) before being coupled to the fresh nonlinearities injected by the subsequent drive ($u_k \mathbf{L}_1$). As a result, Hamiltonian encoding uniquely exploits non-commutativity to autonomously synthesize a highly complex, non-separable temporal hierarchy, equipping the reservoir with the genuine cross-temporal nonlinearities required for advanced temporal processing.

\subsection{Quantifying Expressive Hierarchy via IPC}
\label{subsec:information_processing_capacity}

Theorem~\ref{thm:hamiltonian_pauli_expansion_main} formally establishes that Hamiltonian encoding overcomes the structural bottleneck of qubit resetting, liberating the reservoir's input dependence from any fixed, finite polynomial order. The benchmark tasks evaluated above—STM, PC, and NARMA—provide essential task-level diagnostics, demonstrating how reservoir performance depends on memory retention, nonlinear temporal processing, and dynamics-generated quantum magic in concrete learning scenarios. However, these benchmarks probe selected target functions, and their specific numerical scores naturally depend on task-specific parameters such as the target definition, delay window, input distribution, and relevant hyperparameters.

We complement these task-level benchmarks with the information processing capacity (IPC) framework~\cite{Dambre2012Information}, which provides a task-independent metric of the reservoir's nonlinear temporal processing ability. The IPC quantifies the reservoir's ability to reconstruct a complete orthogonal set of functions of past inputs, making it a natural empirical counterpart to Theorem~\ref{thm:hamiltonian_pauli_expansion_main}: Hamiltonian encoding should access a strictly broader hierarchy of temporal features than qubit resetting, not merely outperform it on individual tasks.

Following the standard IPC formalism, the target functions are chosen as a complete orthogonal set of Legendre-polynomial functions,
\begin{equation}
    \overline{y}_k = \prod_i P_{d_i}\!\left(\tilde{u}_{k-\tau_i}\right),
    \label{eq:legendre}
\end{equation}
where $\tilde{u}_k$ is sampled uniformly from the interval $[-1,1]$, and $P_{d_i}$ denotes the normalized Legendre polynomial of degree $d_i$. The total degree of the target function is defined as $d = \sum_i d_i$. 
For convenience, the physical input variable driving the reservoir dynamics is rescaled to $u_k = (\tilde{u}_k+1)/2$, ensuring $u_k \in [0,1]$. For a given target sequence $\overline{y}_k$, the corresponding IPC is defined as:
\begin{equation}
    C_{\mathrm{IPC}}(\overline{y}_k) = 1 - \frac{\frac{1}{K}\sum_{k=1}^{K}\left(y_k-\overline{y}_k\right)^2}{\frac{1}{K}\sum_{k=1}^{K}\left(\overline{y}_k\right)^2},
\end{equation}
where $y_k$ is the output produced by the reservoir readout and $K$ is the length of the evaluation sequence. This metric precisely quantifies the reservoir's reconstruction accuracy, with $C_{\mathrm{IPC}}=1$ indicating perfect prediction.

For each total degree $d$, the degree-resolved total IPC reported in Fig.~\ref{Fig_3} is defined as
\begin{equation}
C_{\mathrm{TOT}}(d)
=
\sum_{\overline{y}_k\in\mathcal{Y}^{\ast}_d}
C_{\mathrm{IPC}}(\overline{y}_k).
\end{equation}
Here, $\mathcal{Y}^{\ast}_d$ denotes the degree-$d$ subset of Legendre target functions, of the form in Eq.~\ref{eq:legendre}, that is used for the IPC evaluation in Fig.~\ref{Fig_3}. For $d=1$, it contains the standard linear memory targets. For $d>1$, we include only target functions with at least one genuinely nonlinear factor, namely at least one $d_i>1$. For example, at $d=2$, the full Legendre basis contains both multilinear memory terms such as
$P_{1}\!\left(\tilde{u}_{k-\tau_1}\right)P_{1}\!\left(\tilde{u}_{k-\tau_2}\right)$
and genuine nonlinear terms such as
$P_{2}\!\left(\tilde{u}_{k-\tau}\right)$.
The former can arise from linear memory mixing across different time steps, whereas the latter requires nonlinear dependence on an individual delayed input. This filtering isolates the nonlinear response that cannot be generated by a linearly encoded qubit-resetting reservoir.

\begin{figure}[tp]
\centering
\includegraphics[width=0.48\textwidth]{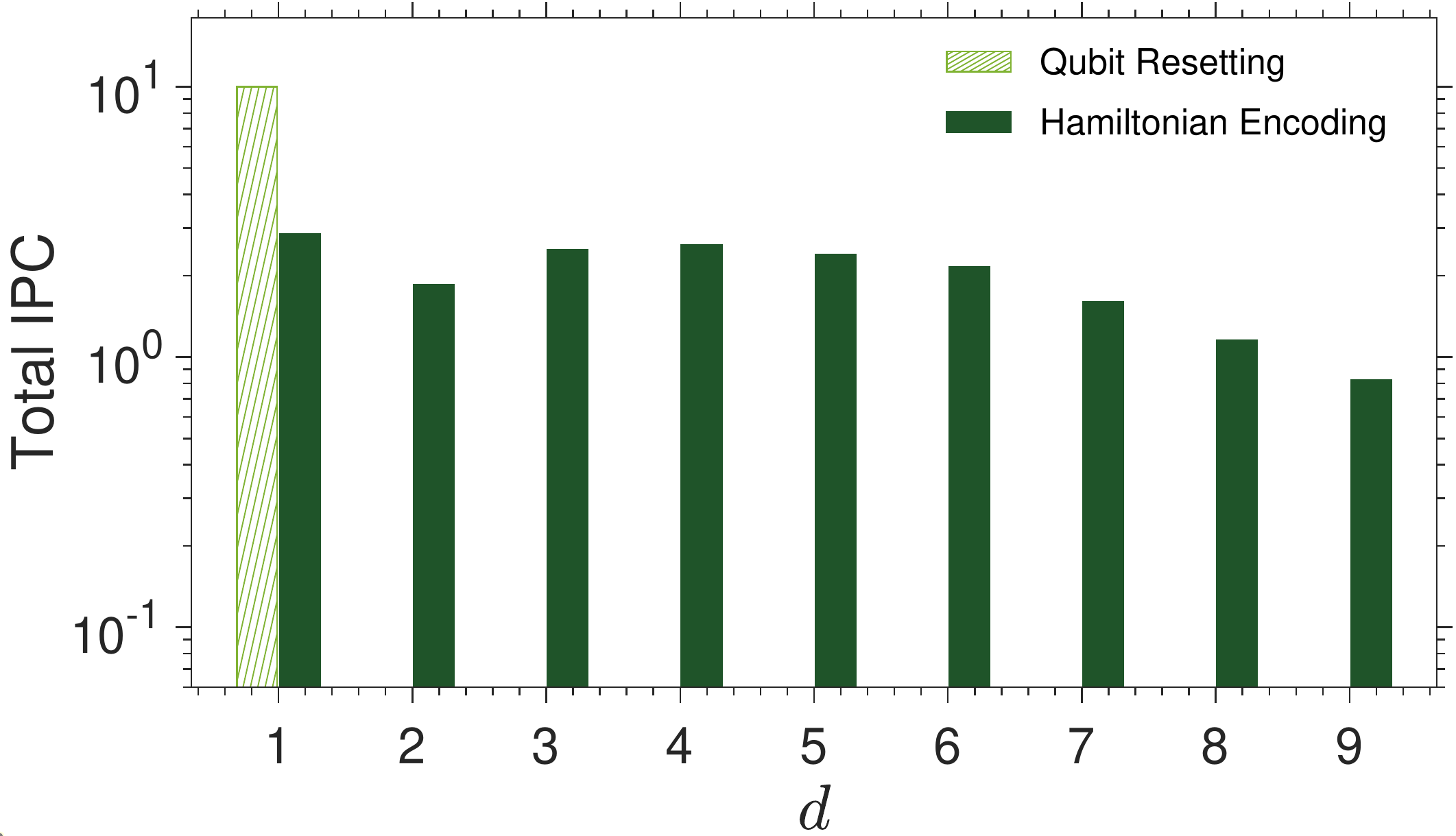}
\caption{
\textbf{
Degree-resolved IPC for qubit resetting and Hamiltonian encoding.} 
For each total degree $d$, the hatched and solid bars denote the total IPC obtained using the qubit-resetting and Hamiltonian-encoding architectures, respectively. Crucially, for data at $d>1$, only target functions containing at least one factor with degree $d_i>1$ are included (see text). This rigorously isolates the intrinsic single-time nonlinear response from mere linear memory mixing. The simulations are governed by the intrinsic many-body Hamiltonian $H_0$ defined in Eq.~\eqref{eq:H0}, utilizing a reservoir size of $N=5$ qubits, $V=10$ virtual time slots, and physical parameters $J=1$, $h_z=1$, $\Delta t=10$, and $\gamma=0.1$. For both encoding architectures, the reported IPC values are ensemble-averaged over 200 independent disorder realizations of the random couplings $J_{ij}$. The washout, training, and testing phases consist of 1000, 3000, and 100 discrete time steps, respectively. 
}
\label{Fig_3}
\end{figure}

To objectively compare the two encoding paradigms, we utilize a standard quantum many-body spin model. The intrinsic unperturbed Hamiltonian is given by~\cite{Fujii2017Harnessing,Sannia2024Dissipation}:
\begin{equation}
    H_0 = \sum_{i < j}^{N} J_{ij}\, \sigma_i^x \sigma_j^x + h_z \sum_{i=1}^{N} \sigma_i^z , 
    \label{eq:H0}
\end{equation}
where the indices label the $N$ reservoir spins, $h_z$ represents the transverse field, and the pairwise couplings $J_{ij}$ are drawn independently from a uniform distribution over $[-J/2, J/2]$.

For the qubit-resetting architecture, while conventional approaches often employ the nonlinear state encoding described in Eq.~\eqref{eq:state_encoding}, our objective here is to strictly isolate the intrinsic nonlinearity generated by the reservoir dynamics. Therefore, we deliberately adopt a strictly linear input encoding: $\rho^{(1)}_{\rm in}(u_k) = \frac{1}{2}(I + u_k Z)$. This guarantees that any higher-order nonlinear response observed in the IPC evaluation originates entirely from the many-body quantum dynamics, rather than being trivially inherited from the classical encoding map. 
For the Hamiltonian encoding architecture, the temporal input directly modulates a driving term $H_{1}(u_k) = h_x(u_k)\sum_{i=1}^{N}\sigma_i^x$, where the amplitude is linearly mapped as $h_x(u_k) = (u_k+1)h_z$. The system's dissipation is governed by local jump operators $L_i = \sigma_i^- = (\sigma_i^x - i\sigma_i^y)/2$, with a uniform decay rate $\gamma_i = \gamma$ applied across all spins. The input drive is held constant over a fixed interval $\Delta t$, enabling a direct, normalized comparison between the discrete-resetting and continuous-driving paradigms.

The evaluated degree-resolved IPC distributions are presented in Fig.~\ref{Fig_3}. Under strictly linear input encoding, the qubit-resetting reservoir has capacity only in the linear sector, while its filtered higher-degree IPC vanishes. In contrast, the Hamiltonian-encoded reservoir retains nonzero capacity across multiple higher degrees, up to $d=9$ in our simulations. This agrees with Theorem~\ref{thm:expressivity_limit} and Theorem~\ref{thm:hamiltonian_pauli_expansion_main}: qubit resetting only linearly mixes the injected input features, whereas Hamiltonian encoding generates intrinsic higher-order nonlinear responses through the input-dependent continuous-time generator. Therefore, the IPC hierarchy in Fig.~\ref{Fig_3} directly highlights the expressivity advantage of Hamiltonian encoding over qubit resetting.

\section{Discussion and Conclusion}
\label{sec:conclusion}

In this work, we have developed a mathematical framework to analyze the fundamental capabilities and limitations of QRC. By mapping the dissipative quantum dynamics into the Pauli-Liouville operator space, we resolved the heuristic ambiguity surrounding the roles of quantum magic, the ESP, and nonlinear expressivity. Our findings establish essential theoretical principles for the design and evaluation of temporal quantum machine learning architectures.

Our theoretical analysis exposes a fundamental dichotomy between two distinct input-encoding paradigms. For the widely utilized qubit-resetting scheme, we rigorously proved a ``no-go'' theorem: its mathematical isomorphism to a classical state-affine system bounds its nonlinear expressive power strictly by the classical injection dimension ($m$). In this architecture, unitary quantum dynamics act merely as linear mixing coefficients. While dynamics-generated quantum magic is essential to drive genuine operator scrambling and sustain a function input-history memory, it cannot elevate the nonlinear polynomial degree of the output. Consequently, the qubit-resetting paradigm suffers from an unavoidable structural trade-off: achieving higher nonlinearity explicitly demands sacrificing temporal memory capacity. More importantly, the theorem implies that this architecture is classically simulable at the level of input–output function classes and fundamentally incapable of yielding genuine function advantage in temporal processing.

To overcome this structural bottleneck, we rigorously analyzed the Hamiltonian encoding scheme. By embedding the temporal input directly into the generator of the open quantum system's dynamics, this architecture entirely decouples the ESP from the generation of quantum magic. The ESP is instead natively guaranteed by the Liouvillian spectral gap of the dissipative environment. Crucially, by embedding the signal directly into the matrix exponential, the discrete-time update map assumes a fundamentally transcendental function form. This autonomously generates an unbounded hierarchy of nonlinear response terms, ensuring an infinite-order nonlinear dependence on the instantaneous drive amplitude. Furthermore, the intrinsic non-commutativity of the open-system generators strictly governs the temporal coupling of these nonlinearities across different time steps, synthesizing a complex, non-separable processing of the input history. 

The Hamiltonian encoding scheme can be implemented with both analogue quantum simulators and digital quantum circuits. On the analogue devices, the input signal is injected through time-dependent Hamiltonian engineering. On fault-tolerant digital platforms, the continuous-time Hamiltonian evolution is implemented via Trotterization. Unlike digital quantum simulation, which demands small Trotter steps to suppress approximation error, QRC requires only that the Trotterized circuit preserve the non-commutativity between input-dependent and input-independent generator terms—making the scheme inherently more robust to coarse digitization.

In summary, we have established a rigorous framework for identifying when quantum reservoir computing offers a genuine computational advantage for temporal information processing. Qubit-resetting encoding is classically simulable at the level of input–output function classes regardless of whether magic is present in the input, the dynamics, or both: the expressivity bottleneck resides in the encoding map, not in the reservoir. Hamiltonian encoding circumvents this bottleneck—non-commutativity between the input-dependent and input-independent terms of the generator is the essential resource that grants access to function classes beyond classical state-affine systems. 

More generally, these results supply a concrete criterion for distinguishing quantum machine learning protocols that are classically simulable from those that are not—a prerequisite for any rigorous claim of quantum advantage in learning tasks.  Architectures whose input-output maps are described by finite-degree polynomials are necessarily classically simulable at the level of function classes, regardless of the quantum resources employed internally. 
By contrast, the Hamiltonian-encoded map cannot, in general, be evaluated efficiently classically unless $\mathsf{BQP}=\mathsf{BPP}$, since the quantum reservoir dynamics in Eq.~\eqref{eq:hamiltonian_discrete_map}, under a general Hamiltonian encoding scheme, can encode any problem in $\mathsf{BQP}$, as shown in Appendix~\ref{app:BQP}. 
This dichotomy suggests a broader principle for dissipative quantum machine learning: representational advantage should be diagnosed not by the size of the Hilbert space, but by the complexity of the encoding-induced input-output map.

\section*{Acknowledgements}
This work is supported by National Program on Key Basic Research Project of China (Grant No. 2021YFA1400900), the Innovation Program for Quantum Science and Technology of China (Grant No. 2024ZD0300100), Shanghai Municipal Science and Technology Major Project (Grant No. 26LZ0500100, 25TQ003, 24LZ1400900, 24LZ1401600, 24DP2600100). W.X. acknowledges support from the National Natural Science Foundation of China/Hong Kong RGC Collaborative Research Scheme (Project CRS CUHK401/22) and New Cornerstone Science Foundation, Hong Kong RGC Senior Research Fellow Scheme, Ref. SRFS2223-4S01.  

\bibliography{references}

\newpage
\appendix

\section{Derivation of the Qubit-Resetting Map in Liouville Space}
\label{app:reset_derivation}

In this appendix, we provide a detailed derivation of the linear resetting matrix $\mathcal{R}(u_k)$ introduced in Section~\ref{subsec:sas_reduction}. Our goal is to rigorously map the physical resetting operation onto a direct matrix-vector multiplication within the Pauli-Liouville operator space.

\subsection{Lexicographical Ordering of the Pauli Basis}
To systematically analyze the operation, we first establish a lexicographical ordering for the $N$-qubit Pauli strings $B^N_i \in \mathcal{P}_N$. We order the single-qubit Pauli matrices as $I < X < Y < Z$ (indexed as $0, 1, 2, 3$, respectively) and compare the tensor factors from left to right, treating the first qubit as the most significant tensor factor. 
Under this strict ordering, any $N$-qubit Pauli string $B^N_i$ can be uniquely factorized into the tensor product of an operator acting on the first qubit and an operator acting on the remaining $N-1$ qubits:
\begin{equation}
    B^N_i = \sigma_{p(i)} \otimes B^{N-1}_{q(i)}, \quad i \in [4^N]_0,
\end{equation}
where the indices are uniquely determined by the quotient and remainder of $c = 4^{N-1}$: 
\begin{equation}
    p(i) := \lfloor i/c \rfloor \in [3]_0, \quad q(i) := i \bmod c \in [c]_0.
\end{equation}

\subsection{Evaluating the Post-Reset State Components}
Let $r_{k-1}$ denote the Pauli-Liouville state vector prior to the resetting operation, and $r'_k$ denote the state vector immediately after (i.e., the post-reset state). According to the definition of the Pauli-Liouville representation [Eq.~\eqref{eq:liouville_state}], the $i$-th component of the updated state is given by:
\begin{equation}
    r'_{k,i} = \frac{1}{\sqrt{2^N}} \Tr\left[ B^N_i \Big( \rho^{(1)}_{\mathrm{in}}(u_k) \otimes \Tr_1[\rho_{k-1}] \Big) \right].
\end{equation}
Substituting the factorized form of the Pauli string $B^N_i = \sigma_{p(i)} \otimes B^{N-1}_{q(i)}$, and utilizing the property that the trace is strictly multiplicative over tensor products, we can decouple the operation into the two distinct subsystems:
\begin{equation}
    r'_{k,i} = \frac{1}{\sqrt{2^N}} \Tr\left[ \sigma_{p(i)} \rho^{(1)}_{\mathrm{in}}(u_k) \right] \cdot \Tr\left[ B^{N-1}_{q(i)} \Tr_1[\rho_{k-1}] \right].
\label{eq:app_trace_split}
\end{equation} 

\subsection{The Partial Trace Identity}
To evaluate the second term in Eq.~\eqref{eq:app_trace_split}, we invoke a fundamental property of the partial trace. For any bipartite state $\rho_{12}$ and an arbitrary operator $B$ acting solely on subsystem 2, the following identity holds:
\begin{equation}
    \Tr \Big[ B \cdot \Tr_1[\rho_{12}] \Big] = \Tr \Big[ (I_1 \otimes B) \rho_{12} \Big].
\end{equation}
Identifying the reset qubit as subsystem 1 and the remaining $N-1$ qubits as subsystem 2, we apply this identity to elevate the trace over the reduced $(N-1)$-qubit subspace to a trace over the full $N$-qubit Hilbert space:
\begin{equation}
    \Tr\left[ B^{N-1}_{q(i)} \Tr_1[\rho_{k-1}] \right] = \Tr \left[ (I \otimes B^{N-1}_{q(i)}) \rho_{k-1} \right].
\end{equation}
Crucially, by virtue of our chosen lexicographical ordering, the tensor product $I \otimes B^{N-1}_{q(i)}$ corresponds exactly to the $N$-qubit Pauli string $B^N_{q(i)}$. This equivalence allows us to directly map this term back to the Pauli-Liouville components of the pre-reset state vector $r_{k-1}$:
\begin{equation}
    \frac{1}{\sqrt{2^N}} \Tr \left[ B^N_{q(i)} \rho_{k-1} \right] = r_{k-1,q(i)}.
\end{equation} 

\subsection{Evaluating the Input Qubit Trace}
We now evaluate the first term in Eq.~\eqref{eq:app_trace_split}. The state of the newly injected input qubit is parameterized by its generalized Bloch vector $\bm{s}_k(u_k) = (1, s_{k,1}, s_{k,2}, s_{k,3})$:
\begin{equation}
    \rho^{(1)}_{\mathrm{in}}(u_k) = \frac{1}{2} \sum_{\mu=0}^{3} s_{k,\mu} \sigma_\mu,
\end{equation}
where $\sigma_0 \equiv I$. Because the Pauli matrices are orthogonal under the trace inner product ($\Tr[\sigma_\mu \sigma_\nu] = 2\delta_{\mu\nu}$), tracing this input state against the single-qubit Pauli operator $\sigma_{p(i)}$ simply extracts the corresponding component of the input vector: $\Tr [ \sigma_{p(i)} \rho^{(1)}_{\mathrm{in}}(u_k) ] = s_{k, p(i)}$.

\subsection{Constructing the Matrix Representation}
Substituting these evaluated components back into Eq.~\eqref{eq:app_trace_split}, the state update rule simplifies remarkably to a component-wise scalar multiplication: $r'_{k,i} = s_{k, p(i)} r_{k-1, q(i)}$. 
We can now translate this component-wise mapping into the global matrix transformation $\mathcal{R}(u_k)$. Recall that $p(i) \in \{0, 1, 2, 3\}$ indicates which macro-block of size $c$ the index $i$ belongs to, and $q(i) \in \{0, \dots, c-1\}$ serves as the local index within that block. 

Because $q(i)$ maps exactly to the first $c$ components of the vector $r_{k-1}$ (the sector where the first qubit is $I$), the updated state component $r'_{k,i}$ depends \textit{exclusively} on the first $c$ components of $r_{k-1}$. The remaining $3c$ components of $r_{k-1}$ (where the first qubit is $X, Y$, or $Z$) are entirely annihilated by the partial trace. The coefficient $s_{k, p(i)}$ is uniformly applied across each respective block $p$. Thus, writing $r'_k = \mathcal{R}(u_k) r_{k-1}$, the matrix $\mathcal{R}(u_k)$ must take the sparse block-column form:
\begin{equation}
    \mathcal{R}(u_k) = \begin{pmatrix}
     I_c & \bm{0} & \bm{0} & \bm{0} \\
     s_{k,1} I_c & \bm{0} & \bm{0} & \bm{0} \\
     s_{k,2} I_c & \bm{0} & \bm{0} & \bm{0} \\
     s_{k,3} I_c & \bm{0} & \bm{0} & \bm{0}
    \end{pmatrix}, \label{app:resetting_matrix}
\end{equation}
where $I_c$ is the $c \times c$ identity matrix. This formally completes the derivation of Eq.~\eqref{eq:reset_matrix} and mathematically proves that qubit resetting acts as a structurally constrained, linear broadcast map in Liouville space.

\section{Derivation of the State-Affine System Reduction}
\label{app:sas_derivation}

In this appendix, we provide a detailed derivation of the exact mathematical isomorphism between the quantum qubit-resetting protocol and the classical SAS introduced in Section~\ref{subsec:sas_reduction}. 
Part of the calculations in this section are inspired by previous work ~\cite{pena2023quantum}. 
Our objective is to analytically reduce the full $4^N$-dimensional Pauli-Liouville dynamics into a compact, closed-loop recurrence relation that governs the $N-1$ un-reset qubits.

\subsection{Block Decomposition of the Operator Space}
We begin with the full state update equation describing a single QRC time step:
\begin{equation}
    r_k = \mathcal{O} \mathcal{R}(u_k) r_{k-1},
\label{eq:app_full_update}
\end{equation}
where $r_k \in \mathbb{R}^{4^N}$ is the Pauli-Liouville state vector, $\mathcal{R}(u_k)$ represents the resetting of the first qubit, and $\mathcal{O} \in \mathrm{SO}(4^N)$ is the orthogonal matrix representing the unitary evolution $U$ with $\mathcal{O}_{ij} = 2^{-N} \Tr[ B^N_i U B^N_j U^\dagger ]$.

To systematically analyze how information flows, we partition the $4^N$-dimensional operator space into four blocks of size $c = 4^{N-1}$. We decompose the state vector $r_k$ as:
\begin{equation}
    r_k = 
    \begin{pmatrix}
        r_k^0 \\
        r_k^1 \\
        r_k^2 \\
        r_k^3
    \end{pmatrix},
\end{equation}
where each sub-vector $r_k^\mu \in \mathbb{R}^c$ corresponds to Pauli strings whose first operator is $\sigma_\mu$ (with $\sigma_0 \equiv I$, $\sigma_1 \equiv X$, $\sigma_2 \equiv Y$, $\sigma_3 \equiv Z$). Thus, $r_k^0$ represents the degrees of freedom strictly local on the $N-1$ un-reset qubits.

Similarly, we partition the $4^N \times 4^N$ unitary evolution matrix $\mathcal{O}$ into sixteen $c \times c$ sub-matrices:
\begin{equation}
    \mathcal{O} =   
    \begin{pmatrix}
        \mathcal{O}^{00} & \mathcal{O}^{01} & \mathcal{O}^{02} & \mathcal{O}^{03} \\
        \mathcal{O}^{10} & \mathcal{O}^{11} & \mathcal{O}^{12} & \mathcal{O}^{13} \\
        \mathcal{O}^{20} & \mathcal{O}^{21} & \mathcal{O}^{22} & \mathcal{O}^{23} \\
        \mathcal{O}^{30} & \mathcal{O}^{31} & \mathcal{O}^{32} & \mathcal{O}^{33}
    \end{pmatrix}.
\end{equation}

\subsection{Extracting the Closed-Loop Recurrence}
We now apply the state update in two stages. First, we apply the resetting matrix $\mathcal{R}(u_k)$ to the previous state $r_{k-1}$. Utilizing the sparse block-column structure of $\mathcal{R}(u_k)$ derived in Appendix~\ref{app:reset_derivation}, the intermediate post-reset state $r'_k$ becomes:
\begin{align}
    r'_k &= \mathcal{R}(u_k) r_{k-1} \nonumber \\
    &= \begin{pmatrix}
     I_c & \bm{0} & \bm{0} & \bm{0} \\
     s_{k,1} I_c & \bm{0} & \bm{0} & \bm{0} \\
     s_{k,2} I_c & \bm{0} & \bm{0} & \bm{0} \\
     s_{k,3} I_c & \bm{0} & \bm{0} & \bm{0}
    \end{pmatrix}
    \begin{pmatrix}
        r_{k-1}^0 \\
        r_{k-1}^1 \\
        r_{k-1}^2 \\
        r_{k-1}^3
    \end{pmatrix} 
    = \begin{pmatrix}
        r_{k-1}^0 \\
        s_{k,1} r_{k-1}^0 \\
        s_{k,2} r_{k-1}^0 \\
        s_{k,3} r_{k-1}^0
    \end{pmatrix}. \label{eq:app_reset_matrix_step}
\end{align}
This explicitly demonstrates that the resetting operation completely annihilates the information in sectors $\mu \in \{1, 2, 3\}$ of the previous state, overwriting them with copies of the un-reset sector $r_{k-1}^0$ weighted by the input components $s_{k,\mu}$.

Next, we apply the unitary evolution matrix $\mathcal{O}$ to this intermediate state to obtain $r_k = \mathcal{O} r'_k$. Because the \textit{subsequent} time step ($k+1$) will again strictly discard all components except $r_k^0$, we only need to calculate the update for the top block $r_k^0$. Multiplying the top row of blocks in $\mathcal{O}$ by the column vector $r'_k$, we obtain:
\begin{align}
    r_k^0 &= \mathcal{O}^{00} r_{k-1}^0 + \mathcal{O}^{01} (s_{k,1} r_{k-1}^0) \nonumber \\
    &\quad + \mathcal{O}^{02} (s_{k,2} r_{k-1}^0) + \mathcal{O}^{03} (s_{k,3} r_{k-1}^0) \nonumber \\
    &= \left( \mathcal{O}^{00} + \sum_{\mu=1}^3 s_{k,\mu} \mathcal{O}^{0\mu} \right) r_{k-1}^0.
\end{align}
Defining $s_{k,0} \equiv 1$, this condenses to a self-contained recurrence relation:
\begin{equation}
    r_k^0 = \left( \sum_{\mu=0}^3 s_{k,\mu} \mathcal{O}^{0\mu} \right) r_{k-1}^0.
\label{eq:app_homogeneous}
\end{equation}
Since $r_k^0$ relies exclusively on its own previous value $r_{k-1}^0$, we have successfully traced out the auxiliary degrees of freedom, creating a closed dynamical loop. To match the notation in the main text, we rename this effective memory vector $x_k := r_k^0$.

\subsection{Isolating the Invariant Trace Component}

Equation~\eqref{eq:app_homogeneous} is a homogeneous linear recurrence. However, we must explicitly account for the physical constraint that valid quantum states possess a unit trace. 
In the Pauli-Liouville representation, the first component of $x_k$ (corresponding to the global identity operator $I^{\otimes N}$) encodes the trace. Because the CPTP map preserves the trace, this component is permanently fixed at $1/\sqrt{2^N}$. We isolate this scalar from the information-carrying (traceless) components by decomposing $x_k$ as:
\begin{equation}
    x_k = \begin{pmatrix} 1/\sqrt{2^N} \\ \tilde{x}_k \end{pmatrix},
\end{equation}
where $\tilde{x}_k \in \mathbb{R}^{c-1}$ represents the expectation values of all non-identity observables on the $N-1$ un-reset qubits.

Because unitary evolution is trace-preserving, the first row of the full matrix $\mathcal{O}$ must strictly be $(1, 0, 0, \dots)$. Furthermore, because unitary evolution is unital (it preserves the identity operator), its first column must be $(1, 0, 0, \dots)^T$. 
Consequently, the top-left element of $\mathcal{O}^{00}$ is exactly 1, and the remainder of its first row and first column are entirely zeros. For $\mu \in \{1, 2, 3\}$, the trace-preserving property dictates that the entire first row of $\mathcal{O}^{0\mu}$ is zero. This imposes the following rigorous block structures on the $c \times c$ sub-matrices:
\begin{equation}
    \mathcal{O}^{00} = \begin{pmatrix} 1 & \bm{0}^T \\ \bm{0} & \tilde{\mathcal{O}}^{00} \end{pmatrix}, \quad 
    \mathcal{O}^{0\mu} = \begin{pmatrix} 0 & \bm{0}^T \\ \tilde{\bm{b}}^\mu & \tilde{\mathcal{O}}^{0\mu} \end{pmatrix},
\end{equation}
where $\tilde{\mathcal{O}}^{0\mu}$ are $(c-1) \times (c-1)$ matrices, and $\tilde{\bm{b}}^\mu$ are $(c-1)$-dimensional column vectors.

Substituting these decomposed forms back into Eq.~\eqref{eq:app_homogeneous}, we can seamlessly separate the trivial trace update from the dynamical update of the traceless components: 
\begin{widetext}
\begin{align}
    \begin{pmatrix} 1/\sqrt{2^N} \\ \tilde{x}_k \end{pmatrix} &= \left[ \begin{pmatrix} 1 & \bm{0}^T \\ \bm{0} & \tilde{\mathcal{O}}^{00} \end{pmatrix} + \sum_{\mu=1}^3 s_{k,\mu} \begin{pmatrix} 0 & \bm{0}^T \\ \tilde{\bm{b}}^\mu & \tilde{\mathcal{O}}^{0\mu} \end{pmatrix} \right]  \begin{pmatrix} 1/\sqrt{2^N} \\ \tilde{x}_{k-1} \end{pmatrix} \nonumber \\[3 pt]
    &= \begin{pmatrix} 1/\sqrt{2^N} \\[6 pt] \frac{1}{\sqrt{2^N}} \sum_{\mu=1}^3 s_{k,\mu} \tilde{\bm{b}}^\mu + \left( \tilde{\mathcal{O}}^{00} + \sum_{\mu=1}^3 s_{k,\mu} \tilde{\mathcal{O}}^{0\mu} \right) \tilde{x}_{k-1} \end{pmatrix}.
\end{align}
\end{widetext} 
Extracting the non-trivial lower block, we obtain an affine transformation for the traceless memory vector:
\begin{equation}
    \tilde{x}_k = \left( \tilde{\mathcal{O}}^{00} + \sum_{\mu=1}^3 s_{k,\mu} \tilde{\mathcal{O}}^{0\mu} \right) \tilde{x}_{k-1} + \frac{1}{\sqrt{2^N}} \sum_{\mu=1}^3 s_{k,\mu} \tilde{\bm{b}}^\mu.
\end{equation}

\subsection{Final State-Affine Form}
To present this elegantly in the full $c$-dimensional space (incorporating the fixed trace component as a bias mechanism), we define a 4-dimensional tensor of matrices $\mathbf{A} = (A_0, A_1, A_2, A_3)$ and a 4-dimensional tensor of vectors $\mathbf{b} = (\bm{b}_0, \bm{b}_1, \bm{b}_2, \bm{b}_3)$, specified as:
\begin{align}
    A_0 &= [1] \oplus \tilde{\mathcal{O}}^{00}, & A_\mu &= [0] \oplus \tilde{\mathcal{O}}^{0\mu}, \\
    \bm{b}_0 &= \bm{0}, & \bm{b}_\mu &= \begin{pmatrix} 0 \\ \frac{1}{\sqrt{2^N}} \tilde{\bm{b}}^\mu \end{pmatrix}.
\end{align}
With these definitions, the complete recurrence relation for the $c$-dimensional effective state $x_k$ condenses exactly to the classical State-Affine System:
\begin{equation}
    x_{k} = (\bm{s}_{k} \cdot \mathbf{A}) \, x_{k-1} + \bm{s}_{k} \cdot \mathbf{b}.
\end{equation}
This rigorously concludes the derivation, mathematically proving that qubit resetting maps strictly onto a classical affine filter, where the encoded input vector $\bm{s}_k$ linearly mixes the internal transition matrices.

\section{Proof of Theorem \ref{thm:esp_without_magic}: Reservoir under Clifford Dynamics} 
\label{app:proof_thm_esp}

In this appendix, we provide a rigorous proof of Theorem~\ref{thm:esp_without_magic}. We begin by establishing the precise conditions under which a Clifford reservoir satisfies the echo state property (ESP). Next, we present a combinatorial derivation demonstrating that, in the thermodynamic limit ($N \to \infty$), a generic Clifford reservoir driven by the input state $\rho^{(1)}_{\rm in}(u_k) = \frac{1}{2}(I + u_k Z)$ satisfies the ESP with a probability of $3/4$. 
Finally, we show that even when the ESP is satisfied, the absence of non-Clifford magic results in an expected number of nonzero Pauli components of order $O(1)$ in the long-time limit.

\subsection{ESP under Clifford Dynamics}

\begin{lemma}[ESP under Clifford Dynamics]\label{lem:esp_under_clifford}
Let $c = 4^{N-1}$. For a Clifford reservoir driven by an input sequence $u_k$, define the projected map $F(i) := f(i) \pmod c$ and the branch map $G_m(i) := \lfloor f(F^{m-1}(i))/c \rfloor$ for $m\in \mathbb{Z}_+$. Let $\mathcal{I}_t = \{ j \in [3]_0 \mid |s_{k,j}| = t, \forall k \}$ for $t \in \{0, 1\}$ denote the sets of indices associated with constant-magnitude input components. For a given output Pauli index $i \in [c]_0 \setminus \{0\}$, let $p$ be the smallest integer such that $F^p(i) = F^q(i)$ for some $q < p$. The output index $i$ satisfies the ESP if and only if at least one of the following conditions holds:
(a) $F^p(i) = 0$;
(b) $\exists\, m \in [q+1, p]$ such that $G_m(i) \notin \mathcal{I}_1$;
(c) $\exists\, m \in \mathbb{Z}_+$ such that $G_m(i) \in \mathcal{I}_0$.
\end{lemma}

\textit{Proof.} Any Clifford unitary $C \in \mathrm{Cl}_N$ induces a strict permutation $f: [4^N]_0 \to [4^N]_0$ and an associated sign function $\eta: [4^N]_0 \to \{+1, -1\}$, such that $C^\dagger B_i^N C = \eta(i) B_{f(i)}^N$. Consequently, the orthogonal matrix $\mathcal{O}$ governing the unitary evolution in the Pauli-Liouville space acts as a signed permutation matrix, containing exactly one nonzero entry ($\pm 1$) per row and column. As shown in Eq.~\eqref{eq:reset_matrix}, the resetting operation distributes the preceding un-reset state $x_{k-1}$ into the intermediate state $r_k'$. The matrix $\mathcal{O}$ then permutes the elements of $r_k'$ to generate the post-unitary state $r_k$. Therefore, using the functions defined in the lemma, the $i$-th Pauli component of the output state $x_k$ depends on exactly one component from the preceding state $x_{k-1}$:
\begin{equation}
    x_{k, i} = \eta(i) \, s_{k, G_1(i)} \, x_{k-1, F(i)},
\end{equation}
where $F(i) := f(i) \pmod c$ dictates the backward trajectory in the un-reset sector, and $G_1(i) := \lfloor f(i)/c \rfloor \in \{0,1,2,3\}$ determines which component of the input vector $\bm{s}_k$ weights the transition.

By recursively unrolling this equation backward in time over $M$ steps, we obtain the explicit dependence of the current state on the deep past:
\begin{equation}
    x_{k, i} = \left[ \prod_{m=1}^{M} \eta(F^{m-1}(i)) \, s_{k-m+1, G_m(i)} \right] x_{k-M, F^M(i)},
\label{eq:backward_unroll}
\end{equation}
where $F^m(i)$ denotes the $m$-th application of the map $F$, and $G_m(i) := \lfloor f(F^{m-1}(i))/c \rfloor$ tracks the input sector targeted at the $m$-th backward step. 

Because the effective state space dimension $c = 4^{N-1}$ is finite, the sequence of indices $\{i, F(i), F^2(i), \dots\}$ cannot wander infinitely without repetition. Topologically, the backward trajectory of any index $i$ must eventually form a ``lasso'' structure: a transient tail of length $q$ followed by a periodic cycle of length $p-q$. Thus, there exist minimal integers $q \ge 0$ and $p > q$ such that $F^p(i) = F^q(i)$.

To determine whether the reservoir satisfies the ESP, we must evaluate the limit of Eq.~\eqref{eq:backward_unroll} as $M \to \infty$. Following this cycle structure, the backward path of the $i$-th Pauli component can be rigorously decomposed as:
\begin{equation} 
    x_{k, i} = W_{0}^{(i)} \left( \prod_{\mu=0}^{\alpha-1} W^{(i)}_{\text{cycle},\mu} \right) x_{k-q-\alpha(p-q), F^q(i)}, \label{eq:app_output_esp}
\end{equation}
where 
\begin{gather}
    W_{0}^{(i)} = \prod_{m=1}^{q} \eta(F^{m-1}(i))\, s_{k-m+1,G_{m}(i)}, \\
    W^{(i)}_{\text{cycle}, \mu} = \prod_{m=q+1}^{p} \eta(F^{m-1}(i))\, s_{k-m-\mu(p-q),G_{m}(i)},
\end{gather}
and $\alpha\in\mathbb{N}$ denotes the number of complete cycles traversed backward. Although the operator index flow is strictly periodic, the temporal inputs vary across different periods, necessitating the distinct cycle index $\mu$ in $W^{(i)}_{\text{cycle},\mu}$. 

As $\alpha \to \infty$, the temporal index $k-q-\alpha(p-q)$ extends into the infinite past. By definition, the Echo State Property (ESP) requires the reservoir's state to asymptotically lose all dependence on its initial condition. According to Eq.~\eqref{eq:app_output_esp}, for a specific index $i$, this implies: 
\begin{equation}
    \text{ESP} \iff \left(F^q(i)=0\right) \vee \left[\lim_{\alpha \to \infty} W_{0}^{(i)} \prod_{\mu=0}^{\alpha-1} W^{(i)}_{\text{cycle},\mu} = 0 \right].
\end{equation}
The first condition corresponds exactly to condition (a) in Lemma~\ref{lem:esp_under_clifford}. When $F^q(i)=0$, the backward trajectory of the index reaches the identity component. Because the trace of any valid physical density matrix is strictly 1, the identity component is invariant: $x_{-q-\alpha(p-q), 0} = 1/\sqrt{2^N}$ for any initial state. Consequently, the state's dependence on the initial condition perfectly vanishes, and the ESP is trivially satisfied. 

To complete the proof, it remains to analyze the vanishing weight scenario. We will demonstrate that conditions (b) and (c) defined in the lemma capture this exact requirement, specifically establishing the equivalence:
\begin{equation}
    \text{(b)} \vee \text{(c)} \iff \lim_{\alpha \to \infty} W_{0}^{(i)} \prod_{\mu=0}^{\alpha-1} W^{(i)}_{\text{cycle},\mu} = 0.
    \label{eq:app_esp_condition}
\end{equation}

If condition (b) holds, there exists an $m_0 \in [q+1,p]$ such that $G_{m_0}(i) \notin \mathcal{I}_1$. Therefore, the magnitude of the cycle weight is bounded by the corresponding input component: $\left|W^{(i)}_{\text{cycle},\mu}\right| \le \left|s_{k-m_0-\mu(p-q), G_{m_0}(i)}\right|$. Let $X_{\mu} = \left|s_{k-m_0-\mu(p-q), G_{m_0}(i)}\right|$ for $\mu \in \mathbb{Z}$ be random variables defined on $[0,1]$. If there exists any $\mu$ such that $X_\mu = 0$, then $\prod_{\mu=0}^{\alpha-1} W^{(i)}_{\text{cycle},\mu} = 0$ follows directly. 

If $X_\mu > 0$ for all $\mu \in \mathbb{Z}$, assuming the input sequence is generated from a stationary probability distribution, we can bound the infinite product as: 
\begin{align}
    \left|\lim_{\alpha \to \infty} \prod_{\mu=0}^{\alpha-1} W^{(i)}_{\text{cycle},\mu}\right| &\le \lim_{\alpha \to \infty} \prod_{\mu=0}^{\alpha-1} X_\mu \notag \\
    &= \lim_{\alpha \to \infty} \exp\left(-\sum_{\mu=0}^{\alpha-1} Y_\mu\right) \notag \\
    &= \lim_{\alpha \to \infty} \exp(-\alpha \mathbb{E}[Y_0]), 
\end{align}
where $Y_\mu = -\ln X_\mu \in [0,\infty)$ for $\mu \in \mathbb{Z}$, and $\mathbb{E}[Y_0]$ is the expected value of this random variable. The final step follows rigorously from the Strong Law of Large Numbers. Because $G_{m_0}(i) \notin \mathcal{I}_1$, the magnitude $X_\mu$ is not constantly $1$. Thus, $X_\mu < 1$ occurs with a strictly positive probability, ensuring $Y_\mu > 0$ with nonzero probability. For random time sequences with respect to the product Lebesgue measure, this dictates that $\mathbb{E}[Y_0] > 0$. Consequently, 
\begin{equation}
    \left|\lim_{\alpha \to \infty} \prod_{\mu=0}^{\alpha-1} W^{(i)}_{\text{cycle},\mu} \right| \le \lim_{\alpha \to \infty} \exp(-\alpha \mathbb{E}[Y_0]) = 0. 
\end{equation}
This confirms that condition (b) is strictly sufficient for the right-hand side of Eq.~\eqref{eq:app_esp_condition}. 

If condition (c) holds, the right-hand side of Eq.~\eqref{eq:app_esp_condition} follows trivially, as the requirement $G_m(i) \in \mathcal{I}_0$ forces at least one of the multipliers to be exactly zero. 
Conversely, if neither (b) nor (c) is satisfied, then $\left|W^{(i)}_{\text{cycle},\mu} \right| = 1$ for all $\mu \in \mathbb{Z}$, which yields
\begin{equation}
    \left|\lim_{\alpha \to \infty} W_0^{(i)} \prod_{\mu=0}^{\alpha-1} W^{(i)}_{\text{cycle},\mu} \right| = |W_0^{(i)}| \neq 0. 
\end{equation}
Therefore, the mathematical equivalence in Eq.~\eqref{eq:app_esp_condition} strictly holds. Combined with condition (a), this completes the proof of Lemma~\ref{lem:esp_under_clifford}. 
\hfill $\blacksquare$

\subsection{Combinatorial Derivation of $p_{\rm ESP} = 3/4$}

We now derive the exact probability of this pathological violation occurring for a random Clifford reservoir in the thermodynamic limit. 

For the $Z$-basis encoding defined in Theorem~\ref{thm:esp_without_magic}, the constant-magnitude index set is trivial, $\mathcal{I}_1 = \{0\}$. Therefore, to violate the ESP, there must exist at least one cycle of Pauli strings that is entirely confined to the un-reset sector. Physically, this means there is a subset of Pauli strings strictly of the form $I \otimes B^{N-1}$ that, under the Clifford unitary $C$, exclusively map to other strings of the exact same form $I \otimes B^{N-1}$, thereby completely bypassing the single-qubit reset dissipation. 
The Clifford group acts on the Pauli group equivalently to the symplectic group $\mathrm{Sp}(2N, \mathbb{Z}_2)$ acting on a $2N$-dimensional binary vector space. The subspace of Pauli operators acting as the identity on the reset qubit constitutes a symplectic subspace of co-dimension 2. 

The evaluation of the ESP thus reduces to a classic problem in random matrix theory over finite fields: What is the distribution of purely invariant cycles within a fixed co-dimension 2 subspace under a random symplectic transformation? According to the cycle index theorems~\cite{harary1973graphical} for the symplectic group over $\mathbb{Z}_2$, as $N \to \infty$, the number of cycles $X$ strictly confined to such a subspace converges to a Poisson distribution:
\begin{equation}
    P(X = z) = \frac{\lambda^z e^{-\lambda}}{z!}.
\end{equation}
The Poisson parameter $\lambda$ is determined by the integral over the invariant measure of the subspace constraints, which for a $\mathbb{Z}_2$ symplectic space of co-dimension 2 evaluates exactly to:
\begin{equation}
    \lambda = \sum_{j=1}^{\infty} \frac{1}{j \cdot 2^{2j}} = -\ln\left(1 - \frac{1}{4}\right) = \ln\left(\frac{4}{3}\right) .
\end{equation}
The reservoir satisfies the ESP if and only if \textit{zero} such pathological cycles exist ($X = 0$). Evaluating the Poisson probability mass function for the null event yields the precise theoretical bound:
\begin{equation}
    p_{\rm ESP} = P(X = 0) = e^{-\lambda} = \frac{3}{4}.
\end{equation}
Thus, a random Clifford reservoir satisfies the ESP with a probability of $3/4$ in the large-$N$ limit. This fractional stability arises strictly from the topological truncation of the backward permutation paths.
We provide direct numerical results in Fig.~\ref{Fig_app1} to quantitatively corroborate the theoretical result. 

\begin{figure}[htbp]
\centering
\includegraphics[width=0.45\textwidth]{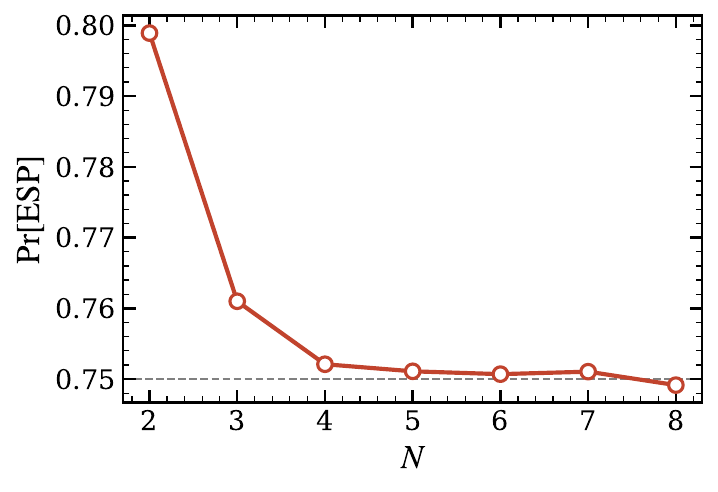}
\caption{
\textbf{Probability of a random Clifford reservoir exhibiting the ESP.} 
The input state is $\rho_{\text{in}}^{(1)}(u_k)=(I+u_kZ)/2$. 
The quantum reservoir consists of a depth-10 random Clifford circuit, where $N$ denotes the system size. The probability is averaged over $10,000$ random Clifford realizations for $N \in \{2, 3, 4, 5, 6\}$, while for $N=7$ and $8$, it is averaged over $20,000$ and $60,000$ realizations, respectively. For larger values of $N$, an increased number of realizations is used to ensure convergence.
}
\label{Fig_app1}
\end{figure}

For a different single-qubit input state, the only distinction is that the index set $\mathcal{I}_1$ may vary, thereby altering the co-dimension of the subspace. 
Consequently, the Poisson parameter $\lambda$ takes a different value while remaining of order $\mathcal{O}(1)$ with respect to the system size $N$. 
Thus, the resulting probability may shift to a different $\mathcal{O}(1)$ constant.
For instance, considering the common input state: 
\begin{equation}
\ket{\psi^{(1)}_{\rm in}(u_k)} = \sqrt{1 - u_k} \ket{0} + \sqrt{u_k} \ket{1},
\label{app:state_encoding}
\end{equation}
the index set $\mathcal{I}_1$ is still $\{0\}$ and the probability therefore remains the same.

\subsection{Expected Number of nonzero Components}

Finally, we demonstrate that in the thermodynamic limit, even when a magic-free Clifford reservoir satisfies the ESP, it yields an expected number of nonzero Pauli components of strictly $O(1)$ in the long-time limit. 

Lemma~\ref{lem:esp_under_clifford} establishes that for a Clifford reservoir satisfying the ESP, every output index must fulfill one of three specific conditions. As demonstrated in our proof, conditions (b) and (c) invariably lead to vanishing components in the long-time limit. Consequently, determining the expected number of surviving nonzero components reduces strictly to evaluating the expected number of indices satisfying condition (a) under a random Clifford unitary.



A more intuitive perspective is that, given the reservoir satisfies the ESP, we can safely initialize the system in the maximally mixed state, represented by $x_{-\tau, i}= \delta_{i,0}/\sqrt{2^N}$ for all $i$. 
From this point, we can trace the forward evolution to observe how information flows within the operator space. 
During the initial resetting operation [Eq.~\eqref{app:resetting_matrix}], the component $x_{-\tau, 0}$ is mapped to $x'_{-\tau+1, 0}$ in the first sector (indices $0$ to $c-1$) and to $x'_{-\tau+1, 3c}$ in the fourth sector (indices $3c$ to $4c-1$). 
Consequently, the initial component is preserved while a new nonzero node is spawned at index $3c$.Subsequently, the applied Clifford unitary maps the corresponding Pauli string ($Z \otimes I^{N-1}$) to a random operator. 
In the thermodynamic limit ($N\to \infty$), this operator lands back in the first sector with a probability of $1/4$, thereby generating an active node that survives the subsequent reset. 
If it maps to any other sector, the information is discarded by the next resetting operation, terminating that specific dynamical path. 
When a node successfully survives in the first sector, the subsequent resetting matrix $\mathcal{R}(u_{k+1})$ physically distributes it to both the first and fourth sectors, creating two new active branches. Combining these operations, a single forward iteration starting from an active node yields two surviving branches with a probability of $1/4$, and zero branches with a probability of $3/4$.

Because multiple active nodes can populate a single hierarchical time step, this forward dynamical path forms a stochastic tree governed rigorously by a Galton-Watson branching process~\cite{harris1963theory}. Starting from the root node at index $3c$, each active node independently yields a random number of offspring: with a probability of $1/4$, it spawns two child nodes (advancing the tree by one level), and with a probability of $3/4$, it yields zero offspring, halting its growth. 

Importantly, during the resetting operation at each time step $k$, the component $x_{k-1,0}=1/\sqrt{2^N}$ is mapped to $x_{k,3c}$ via multiplication by $u_k$. 
Provided that the input $u_k \neq 0$ for all $k$, each resetting step effectively spawns a new root node that branches out in subsequent time steps. 
At the output step $k=0$, the post-reset state $x'_0$ is composed of the $l$-th generation of nodes originating from the roots seeded at $x'_{-l+1}$ (for $l=1,2,\dots$), which collectively constitute a complete tree structure. 
Consequently, since the components in $x_0$ are merely a permutation of those in $x'_0$, the total number of nodes in this tree corresponds exactly to the expected number of nonzero, traceless Pauli components in $x_0$. 
Let $N_{\rm nodes}$ denote the total size of this tree. The expectation value $\mathbb{E}[N_{\rm nodes}]$ satisfies the recursive relation:
\begin{equation}
    \mathbb{E}[N_{\rm nodes}] = \frac{3}{4} \times 1 + \frac{1}{4} \big(1 + 2\mathbb{E}[N_{\rm nodes}]\big), 
    \label{app:branching}
\end{equation}
where the term $2\mathbb{E}[N_{\rm nodes}]$ accounts for the expected number of nodes contained within the two structurally identical subtrees generated upon a successful split. 
The average is taken over the Clifford circuit ensemble. 
Solving this relation yields $\mathbb{E}[N_{\rm nodes}] = 2$. Including the zeroth component (the global identity operator), which is permanently fixed to $1/\sqrt{2^N}$ due to the trace-preserving property of the quantum dynamics, the total expected number of nonzero Pauli components is exactly $3$. 
Note that this conclusion relies on the assumption that every input $u_k$ is nonzero. 
If any input within the sequence is zero, certain branches of the tree may be prematurely truncated, resulting in a smaller final expected number of active nodes. 
Therefore, the value 3 serves as a rigorous upper bound for any input sequence.
We provide direct numerical results in Fig.~\ref{Fig_app2} to quantitatively corroborate the theoretical result.

\begin{figure}[htbp]
\centering
\includegraphics[width=0.45\textwidth]{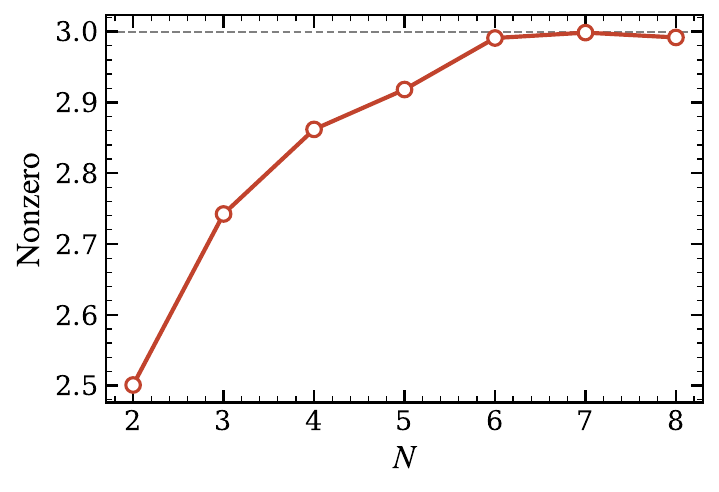}
\caption{
\textbf{Number of the nonzero outputs of a random Clifford reservoir exhibiting the ESP.} 
The input state is $\rho_{\text{in}}^{(1)}(u_k)=(I+u_kZ)/2$. 
The quantum reservoir consists of a depth-10 random Clifford circuit, where $N$ denotes the system size. The number of nonzero outputs is averaged over $10,000$ random Clifford realizations for $N \in \{2, 3, 4, 5, 6\}$, while for $N=7$ and $8$, it is averaged over $20,000$ and $60,000$ realizations, respectively. For larger values of $N$, an increased number of realizations is used to ensure convergence.
}
\label{Fig_app2}
\end{figure}

For a different single-qubit input state, the only distinction is that the index set $\mathcal{I}_0$ may vary, thereby altering the branching number at each level. 
When $|\mathcal{I}_0|\neq0$, the coefficient before $\mathbb{E}[N_{\rm nodes}]$ on the right hand side of Eq.~\eqref{app:branching} may vary, while the solution remains of order $O(1)$ with respect to the system size $N$. 
Thus, the resulting expectation number also shifts to a different $O(1)$ constant. 
For instance, considering the common input state in Eq.~\eqref{app:state_encoding}, we have $|\mathcal{I}_0|=1$. Following a calculation similar to the one above, the total expected number of nonzero Pauli components is exactly $9$.

When $|\mathcal{I}_0|=0$, $\mathbb{E}[N_{\rm nodes}]$ tends to diverge. 
However, given an arbitrarily small numerical precision threshold $\epsilon > 0$, we will demonstrate through a refined analysis that the expected number of nonzero Pauli components with magnitudes exceeding $\epsilon$ remains $O(1)$, independent of the system size $N$.

Suppose $|\mathcal{I}_1|=1$. Consequently, the remaining three channels, corresponding to $f(F^{q-1}(i)) \in \{c, 2c, 3c\}$, are dynamically coupled to the input $u_k$ and are not constant. 
The average magnitude of each channel is a constant within $(0,1)$. Without loss of generality, we set this magnitude to $1/2$ (the subsequent proof holds universally for any arbitrary constant in this interval). 
Recalling that the resetting operation is governed by the matrix in Eq.~\eqref{app:resetting_matrix}, this dynamical evolution maps exactly to a Galton-Watson-type branching process defined on a valued tree. 
Building upon this tree structure described above, we assign a scalar value to each node. Starting from three root nodes at index $\{c,2c,3c\}$, each active node with value $v$ independently yields a random number of offspring: with a probability of 1/4, it spawns one child node with value $v$ and three nodes with value $v/2$ (advancing the tree by one level), and with a probability of 3/4, it yields zero offspring, halting its growth. 
To evaluate the tree originating from a single root node with an initial value of $1$, let $N_i$ denote the number of nodes carrying exactly the value $2^{-i}$, and let $e_i = \mathbb{E}[N_i]$ denote its expectation. 
We define a generalized function $E(V, x)$ representing the expected number of nodes with value $x$ generated from a root node with initial value $V$. By definition, the target expectation is $e_i = E(1, 2^{-i})$. 
Because the branching rules scale proportionally with the parent node's value, the function $E$ exhibits strict scale invariance:
\begin{equation}
    E(V, x) = E\left(1, \frac{x}{V}\right).
\end{equation}
For the base case $i=0$, a value of $1$ can only be inherited directly:
\begin{equation}
    e_0 = 1 + \frac{1}{4} \Big[ 1 \cdot E(1, 1) \Big] + \frac{3}{4} \cdot 0 = 1 + \frac{1}{4}e_0,
\end{equation}
which immediately gives $e_0 = 4/3$.
For $i > 0$, the initial root node does not contribute directly to the count, leading to the recurrence relation:
\begin{align}
    e_i &= 0 + \frac{1}{4} \Big[ 1 \cdot E(1, 2^{-i}) + 3 \cdot E\left(\frac{1}{2}, 2^{-i}\right) \Big] \notag \\
    &= \frac{1}{4} \Big[ e_i + 3 \cdot E\left(1, 2^{-i+1}\right) \Big] \notag \\
    &= \frac{1}{4} e_i + \frac{3}{4} e_{i-1}.
\end{align}
Solving this relation yields $e_i = e_{i-1}$, implying $e_i = 4/3$ for all $i \in \mathbb{N}$.

Finally, given an arbitrarily small threshold $\epsilon \in (0,1)$, the expected number of nodes with values strictly greater than $\epsilon$ is
\begin{equation}
    \mathbb{E}[N_{> \epsilon}] = \frac{4}{3}(\lfloor -\log_2 \epsilon \rfloor + 1), 
\end{equation}
which is of order $O(1)$ with respect to the system size $N$.

\section{Proof of Theorem \ref{thm:expressivity_limit}: Expressivity Limit of Qubit Resetting}
\label{app:proof_thm_expressivity}

In this appendix, we provide a rigorous mathematical proof of Theorem~\ref{thm:expressivity_limit}. Our objective is to demonstrate that an $N$-qubit reservoir, where $m$ qubits are reset at each time step using standard $Z$-basis encoding, is fundamentally isomorphic to a classical polynomial filter. Specifically, we prove that the dependence of any output observable on a specific past input $u_k$ is strictly bounded by a polynomial of degree $m$.

\subsection{Pauli Expansion of the $m$-Qubit Input State}
We begin by analyzing the exact function form of the injected quantum state. In the generalized $m$-qubit resetting scheme, the first $m$ qubits are traced out and replaced by $m$ independent copies of the input-encoded state:
\begin{equation}
    \rho^{(m)}_{\rm in}(u_k) = \left[ \frac{1}{2}(I + u_k Z) \right]^{\otimes m}.
\end{equation}
To understand how this state acts in operator space, we expand this $m$-fold tensor product into the local Pauli basis $\mathcal{P}_m = \{I, X, Y, Z\}^{\otimes m}$. Using the lexicographical ordering of Pauli strings, the generalized Bloch vector representing this $m$-qubit input, denoted as $\bm{s}^{(m)}_k(u_k) \in \mathbb{R}^{4^m}$, is given by
\begin{equation}
    \bm{s}^{(m)}_k(u_k) = (1, 0, 0, u_k)^{\otimes m}. 
\end{equation}
Let $w_Z(\mu)$ denote the Hamming weight (the number of $Z$ operators) of the $\mu$-th Pauli string $P_\mu \in \mathcal{P}_m$. The $\mu$-th component of the Bloch vector scales as $s^{(m)}_{k,\mu}(u_k) \propto u_k^{w_Z(\mu)}$. For example, the identity string $I^{\otimes m}$ has $w_Z = 0$, while the fully polarized string $Z^{\otimes m}$ has $w_Z = m$. Because the maximum possible number of $Z$ operators in an $m$-qubit string is exactly $m$, the generalized Bloch vector $\bm{s}^{(m)}_k(u_k)$ is a vector polynomial in $u_k$ whose maximum degree is strictly bounded by $m$.

\subsection{Generalizing the State-Affine System}
We now trace out the $m$ reset qubits to track the effective memory state of the remaining $N-m$ qubits. Let $c_m = 4^{N-m}$ denote the dimension of the un-reset Liouville subspace. 

Following the identical logical reduction utilized in Appendix~\ref{app:sas_derivation}, the state vector of the un-reset qubits, $x_k \in \mathbb{R}^{c_m}$, updates according to a linear combination of internal transition matrices. However, instead of four branches (as in the single-qubit reset case), the dynamics are now driven by the $4^m$ components of the $m$-qubit input vector:
\begin{equation}
    x_k = \left[ \sum_{\mu=0}^{4^m-1} s^{(m)}_{k,\mu}(u_k) A_\mu \right] x_{k-1} + \sum_{\mu=0}^{4^m-1} s^{(m)}_{k,\mu}(u_k) \bm{b}_\mu ,
\end{equation}
where $A_\mu \in \mathbb{R}^{c_m \times c_m}$ and $\bm{b}_\mu \in \mathbb{R}^{c_m}$ are constant matrices and vectors determined entirely by the fixed unitary evolution $U$.

We can compactly rewrite this $m$-qubit resetting map by defining the matrix tensor $\mathbf{A} = (A_0, A_1, \dots, A_{4^m-1})$ and the vector tensor $\mathbf{b} = (\bm{b}_0, \bm{b}_1, \dots, \bm{b}_{4^m-1})$:
\begin{equation}
    x_k = \big(\bm{s}^{(m)}_k \cdot \mathbf{A}\big) x_{k-1} + \bm{s}^{(m)}_k \cdot \mathbf{b}.
\label{eq:app_m_sas}
\end{equation}
This formulation perfectly mirrors the structure of Eq.~\eqref{eq:sas_framework}. The critical distinction is that the tensors $\bm{s}^{(m)}_k, \mathbf{A}$, and $\mathbf{b}$ now contain $4^m$ components, and the input vector $\bm{s}^{(m)}_k(u_k)$ consists of monomials in $u_k$ with degrees strictly bounded by $m$ rather than $1$. 

\subsection{Unrolling the Volterra Series}
To determine the expressive power of the reservoir, we must evaluate how the output observable at time $k=0$, denoted $x_0$, depends on the infinite sequence of past inputs $\{u_0, u_{-1}, u_{-2}, \dots\}$. 

Let $C_k := \bm{s}^{(m)}_k \cdot \mathbf{A} \in \mathbb{R}^{c_m \times c_m}$ and $D_k := \bm{s}^{(m)}_k \cdot \mathbf{b} \in \mathbb{R}^{c_m}$. Because $\mathbf{A}$ and $\mathbf{b}$ are constant structures, every individual entry of $C_k$ and $D_k$ is a polynomial function of the instantaneous input $u_k$:
\begin{equation}
    C_{k,ij} \in \mathbb{R}[u_k], \quad D_{k,i} \in \mathbb{R}[u_k]. 
\end{equation}
Since $\bm{s}^{(m)}_k = (1, 0, 0, u_k)^{\otimes m}$, the degree of each polynomial entry is strictly bounded by $m$. Therefore,  
\begin{equation}
    C_{k,ij} \in \mathbb{R}_m[u_k], \quad D_{k,i} \in \mathbb{R}_m[u_k]. 
\end{equation}

By recursively applying the state update equation [Eq.~\eqref{eq:app_m_sas}], the output vector at $k=0$ evaluates as:
\begin{align}
    x_0 &= C_0 x_{-1} + D_0 \notag \\
    &= \lim_{M \to \infty} \Big( C_0 \dots (C_{-M} x_{-M-1} + D_{-M} ) \dots + D_0 \Big) \notag \\
    &= \lim_{M \to \infty} \left[ \left(\prod_{k=0}^{M} C_{-k}\right) x_{-M-1} + \sum_{k=0}^{M} \left( \prod_{\tau=0}^{k-1} C_{-\tau} \right) D_{-k} \right] \notag \\
    &= \sum_{k=0}^{\infty} \left( \prod_{\tau=0}^{k-1} C_{-\tau} \right) D_{-k}, \label{eq:sup_state}
\end{align}
where the echo state property (ESP) guarantees that the dependence on the initial state from the infinite past ($x_{-M-1}$) vanishes asymptotically. We adopt the standard convention that an empty matrix product $\prod_{\tau=0}^{-1} C_{-\tau}$ yields the identity matrix $I$. 

Because the input variables at different time steps are structurally independent, the matrix multiplications in Eq.~\eqref{eq:sup_state} only mix cross-terms operating on different time indices ($u_{-\tau}$ and $u_{-k}$). Consequently, variables associated with the exact same instantaneous input are never multiplied together. Based on the algebraic properties established above, it strictly follows that each vector component in the infinite sum is a multivariate polynomial in $u_{-k}, \dots, u_0$, where the maximum degree of any single variable is bounded by $m$: 
\begin{equation}
    \left[ \left( \prod_{\tau=0}^{k-1} C_{-\tau} \right) D_{-k} \right]_i \in \mathbb{R}_m[u_{-k}, \dots, u_0], 
\end{equation}
for any index $i \in [c_m]_0$. 

Summing these terms over the entire history yields the strict inclusion: 
\begin{equation}
    x_{0,i} \in \bigcup_{k \ge 0} \mathbb{R}_m[u_{-k}, \dots, u_0] = \mathbb{R}_m[u_{-\infty}, \dots, u_0], 
\end{equation}
for all $i \in [c_m]_0$, which naturally follows from the definition of the bounded-degree function space $\mathbb{R}_m$.

Finally, we note that for mathematical clarity, the preceding analysis uses the first $c_m$ components ($x_k$) of the full Pauli-Liouville state vector $r_k$ as the effective memory state. The main theorems (Theorem~\ref{thm:esp_without_magic} and Theorem~\ref{thm:expressivity_limit}) remain completely valid for the full state $r_k$. The un-reset components $x_{k-1}$ dictate the intermediate state $r_k'$ immediately following the reset, and the subsequent unitary evolution (from $r_k'$ to $r_k$) acts purely as an input-independent linear transformation. This linear scrambling alters neither the asymptotic convergence behavior nor the polynomial degree bound of the underlying function space.

\section{Decomposition of the Stabilizer R\'enyi Entropy under Qubit Resetting}
\label{app:sre_decomposition}

Complementary to the main results shown above, we evaluate the amount of magic injected by the input state $\rho_{\text{in}}(u_k)$ during the qubit resetting operation, i.e., input-generated magic. 
In fact, our mathematical framework is naturally suitable for the description of magic.

\subsection{Additivity of the SRE under Tensor Products}

Immediately after the local resetting operation, the state of the $N$-qubit reservoir takes the form of a bipartite product state:
\begin{equation}
    \rho'_k = \rho^{(1)}_{\rm in}(u_k) \otimes \Tr_1[\rho_{k-1}],
\end{equation}
where $\rho^{(1)}_{\rm in}(u_k)$ acts on the input qubit, and the partial trace acts on the remaining $N-1$ qubits. 

To analyze the magic of this state, we evaluate the 2-R\'enyi SRE, defined as $\tilde{M}_2(\rho) = M_2(\rho) - S_2(\rho)$. We first consider the 2-R\'enyi purity entropy, $S_2(\rho) = -\log \Tr(\rho^2)$. Because the trace operation is strictly multiplicative over tensor products ($\Tr[A \otimes B] = \Tr[A]\Tr[B]$), the purity entropy is intrinsically additive:
\begin{align}
    S_2(\rho'_k) &= -\log \Big( \Tr[(\rho^{(1)}_{\rm in})^2] \cdot \Tr[(\Tr_1[\rho_{k-1}])^2] \Big) \nonumber \\
    &= S_2(\rho^{(1)}_{\rm in}) + S_2(\Tr_1[\rho_{k-1}]).
\end{align} 
Next, we examine the stabilizer overlap term, $M_2(\rho) = -\log \big( d \Tr[Q \rho^{\otimes 4}] \big)$. For an $N$-qubit system ($d = 2^N$), the projector onto the stabilizer code subspace is given by $Q_N = d^{-2} \sum_{P \in \mathcal{P}_N} P^{\otimes 4}$. Because the Pauli group factorizes across subsystems ($\mathcal{P}_N \cong \mathcal{P}_1 \otimes \mathcal{P}_{N-1}$), the stabilizer projector factorizes correspondingly: $Q_N = Q_1 \otimes Q_{N-1}$. 
Substituting this factorized projector into the definition of $M_2$ yields:
\begin{align}
    M_2(\rho'_k) &= -\log \Big[ 2^N \Tr\big[(Q_1 \otimes Q_{N-1}) (\rho^{(1)}_{\rm in} \otimes \Tr_1[\rho_{k-1}])^{\otimes 4}\big] \Big] \nonumber \\
    &= -\log \Big[ \big(2 \Tr[Q_1 (\rho^{(1)}_{\rm in})^{\otimes 4}]\big) \nonumber \\
    &\quad \times \big(2^{N-1} \Tr[Q_{N-1} (\Tr_1[\rho_{k-1}])^{\otimes 4}]\big) \Big] \nonumber \\
    &= M_2(\rho^{(1)}_{\rm in}) + M_2(\Tr_1[\rho_{k-1}]).
\end{align} 
Subtracting the additive $S_2$ term from the additive $M_2$ term, we conclude that the SRE perfectly decouples across the product-state bipartition:
\begin{equation}
    \tilde{M}_2(\rho'_k) = \tilde{M}_2(\rho^{(1)}_{\rm in}(u_k)) + \tilde{M}_2(\Tr_1[\rho_{k-1}]).
    \label{eq:M2}
\end{equation}
The first term quantifies the external magic injected locally by the classical input sequence, while the second term captures the magic intrinsically retained within the correlations of the $N-1$ un-reset qubits, which we refer to as the reservoir memory magic.

\subsection{Analytical Evaluation of Input-Generated Magic}
We now explicitly calculate the first term in Eq.~\eqref{eq:M2}, $\tilde{M}_2(\rho^{(1)}_{\rm in}(u_k))$, which represents the fresh quantum magic injected by the continuous input variable $u_k$. The single-qubit input state is parameterized by its generalized Bloch vector:
\begin{equation}
    \rho^{(1)}_{\rm in}(u_k) = \frac{1}{2} \sum_{\mu=0}^3 s_{k,\mu} \sigma_\mu,
\end{equation}
where $\sigma_0 \equiv I$. 
We evaluate the purity by substituting this expansion into $\Tr(\rho^2)$. Using the orthogonality of Pauli matrices, $\Tr(\sigma_\mu \sigma_\nu) = 2\delta_{\mu,\nu}$, we obtain:
\begin{equation}
    \Tr[(\rho^{(1)}_{\rm in})^2] = \frac{1}{4} \sum_{\mu,\nu=0}^3 s_{k,\mu} s_{k,\nu} \Tr(\sigma_\mu \sigma_\nu) = \frac{1}{2} \sum_{\mu=0}^3 s_{k,\mu}^2.
\end{equation}
Consequently, the purity entropy is:
\begin{equation}
    S_2(\rho^{(1)}_{\rm in}) = -\log \left( \frac{1}{2} \sum_{\mu=0}^3 s_{k,\mu}^2 \right).
\label{eq:app_s2_single}
\end{equation}

Next, we evaluate the stabilizer overlap $\Tr[Q_1 (\rho^{(1)}_{\rm in})^{\otimes 4}]$. Expanding the four-fold tensor product gives:
\begin{equation}
    (\rho^{(1)}_{\rm in})^{\otimes 4} = \frac{1}{16} \sum_{\alpha,\beta,\gamma,\delta} s_{k,\alpha} s_{k,\beta} s_{k,\gamma} s_{k,\delta} \, (\sigma_\alpha \otimes \sigma_\beta \otimes \sigma_\gamma \otimes \sigma_\delta).
\end{equation}
The single-qubit stabilizer projector is $Q_1 = \frac{1}{4} \sum_{\mu=0}^3 \sigma_\mu^{\otimes 4}$. When we trace the product of $Q_1$ and $(\rho^{(1)}_{\rm in})^{\otimes 4}$, the trace distributes over each tensor factor. Because $\Tr(\sigma_\mu \sigma_\alpha) = 2\delta_{\mu,\alpha}$, the only non-vanishing terms in the multiple sum occur when all indices strictly match ($\alpha=\beta=\gamma=\delta=\mu$). For these surviving terms, the trace yields a factor of $2^4 = 16$. Thus, the overlap simplifies exactly to:
\begin{align}
    \Tr[Q_1 (\rho^{(1)}_{\rm in})^{\otimes 4}] &= \frac{1}{64} \sum_{\mu=0}^3 s_{k,\mu}^4 \Tr[\sigma_\mu^{\otimes 4} (\sigma_\mu^{\otimes 4})] \nonumber \\
    &= \frac{16}{64} \sum_{\mu=0}^3 s_{k,\mu}^4 = \frac{1}{4} \sum_{\mu=0}^3 s_{k,\mu}^4.
\end{align}
The magic overlap $M_2$ for this single qubit ($d=2$) is therefore:
\begin{align}
    M_2(\rho^{(1)}_{\rm in}) &= -\log \left( 2 \cdot \Tr[Q_1 (\rho^{(1)}_{\rm in})^{\otimes 4}] \right) \nonumber \\
    &= -\log \left( \frac{1}{2} \sum_{\mu=0}^3 s_{k,\mu}^4 \right).
\label{eq:app_m2_single}
\end{align}

Finally, by subtracting Eq.~\eqref{eq:app_s2_single} from Eq.~\eqref{eq:app_m2_single}, the constants $\log(1/2)$ elegantly cancel out. Combining the logarithmic terms yields the exact expression for the injected magic:
\begin{equation}
    \tilde{M}_2(\rho^{(1)}_{\rm in}(u_k)) = -\log \left( \frac{\sum_{\mu=0}^3 s_{k,\mu}^4}{\sum_{\mu=0}^3 s_{k,\mu}^2} \right).
\end{equation}
This expression explicitly shows how the classical input encoding controls the amount of input-generated magic injected into the reservoir.

\section{Proof of Theorem~\ref{thm:hamiltonian_pauli_expansion_main}: Infinite-order Expressivity of Hamiltonian Encoding}
\label{app:hamiltonian_encoding_proof}

In this appendix, we present the rigorous proof of the main-text result regarding the expressive structure of Hamiltonian encoding. Our objective is to demonstrate that once the continuous temporal input is embedded directly into the generator of the open-system dynamics, the resulting discrete-time update map exhibits a fundamentally transcendental function dependence on the input amplitude. To this end, we first derive the exact Pauli-Liouville space representation of the one-step map from the Lindblad master equation, and then prove that it is transcendental. This establishes the infinite-order nonlinear structure that fundamentally distinguishes Hamiltonian encoding from qubit-resetting paradigms.

Because the dynamical map assumes the exact same function form at each discrete time step, we henceforth suppress the time index $k$ in $r_k$ and $u_k$ whenever the context is clear. For notational simplicity, we also drop the system-size superscript $N$ from the Pauli basis elements $B_i^N$, denoting them simply as $B_i$.

In the first part of the proof, we demonstrate that the discrete-time dynamical map takes the exact form
\begin{equation}
    r = \exp\Big[ \Delta t (\mathbf{L}_0 + u \mathbf{L}_1) \Big] r_{\text{prev}}, 
\label{app:hamiltonian_discrete_map}
\end{equation}
where $\mathbf{L}_0$ and $\mathbf{L}_1$ are real, input-independent matrices, and $\mathbf{L}_1$ possesses at least one nonzero eigenvalue. 
The continuous-time evolution governed by the Lindblad master equation in Eq.~\eqref{eq:lindblad_master} is given by
\begin{equation}
    \frac{d\rho(t)}{dt} 
    = -i \big[H_0 + u H_1, \, \rho(t)\big] + \sum_\mu \gamma_\mu \mathcal{D}[L_\mu] \rho(t). 
\label{app:lindblad_master}
\end{equation}
We map this density matrix to the Pauli-Liouville space by expanding it in the trace-orthogonal Pauli basis, $\rho(t) = \frac{1}{\sqrt{2^N}} \sum_j r_j(t) B_j$. 
Substituting this expansion into the master equation yields
\begin{align}
    &\frac{1}{\sqrt{2^N}} \sum_i \frac{dr_i}{dt} B_i \notag \\
    =& \frac{1}{\sqrt{2^N}} \sum_j r_j \left( -i [H_0 + uH_1, B_j] + \sum_\mu \gamma_\mu \mathcal{D}[L_\mu]B_j \right). \notag
\end{align}
Multiplying both sides by $B_i$ and taking the trace, while utilizing the orthogonality condition $\Tr(B_i B_j)=2^N\delta_{ij}$, we isolate the time derivative of the coefficient vector:
\begin{align}
    \frac{dr_i}{dt} = \sum_j (\mathbf{L}_0 + u\mathbf{L}_1)_{ij} r_j. 
\end{align}
Here, the matrix elements of the time-independent generators are explicitly defined as:
\begin{align}
    (\mathbf{L}_0)_{ij} &= \frac{1}{2^N} \Tr \left[ \left( -i [H_0, B_j] + \sum_\mu \gamma_\mu \mathcal{D}[L_\mu]B_j \right) B_i \right], \notag \\
    (\mathbf{L}_1)_{ij} &= \frac{-i}{2^N} \Tr \left[ [H_1, B_j] B_i \right] = \frac{-i}{2^N} \Tr \left[H_1 [B_j, B_i] \right], \notag
\end{align}
where the cyclic property of the trace was applied to simplify $\mathbf{L}_1$. Integrating this linear system of ordinary differential equations over the constant-input interval $\Delta t$ directly yields Eq.~\eqref{app:hamiltonian_discrete_map}, which corresponds to Eq.~\eqref{eq:hamiltonian_discrete_map} in the main text. 
In what follows, we demonstrate that a Taylor series expansion of this discrete-time update map with respect to the drive amplitude $u$ naturally recovers the infinite power series formulation presented in Theorem~\ref{thm:hamiltonian_pauli_expansion_main}.

We begin by rigorously characterizing the algebraic structure of the control generator $\mathbf{L}_1$. Let $f_{ijk}$ denote the fully anti-symmetric structure constants of the Pauli Lie algebra, defining the Lie bracket as $[B_i, B_j] = 2i \sum_k f_{ijk} B_k$. Substituting this into the expression for $\mathbf{L}_1$, and noting that $[B_j, B_i] = -2i \sum_k f_{ijk} B_k$, we obtain
\begin{equation}
    (\mathbf{L}_1)_{ij} = -\frac{1}{2^{N-1}} \sum_k f_{ijk} \Tr(H_1 B_k) = -\frac{1}{2^{N-1}} \sum_k f_{ijk} h_k, \notag
\end{equation}
where $h_k = \Tr(H_1 B_k)$ are the real expansion coefficients of the drive Hamiltonian $H_1$ in the Pauli basis. Owing to the anti-symmetry of the structure constants under index permutation ($f_{jik} = -f_{ijk}$), it immediately follows that $(\mathbf{L}_1)_{ij} = -(\mathbf{L}_1)_{ji}$. This establishes $\mathbf{L}_1$ as a strictly real skew-symmetric matrix ($\mathbf{L}_1^T = -\mathbf{L}_1$). By the premise of the theorem, the control Hamiltonian is non-trivial ($H_1 \neq c I$ for any real constant $c$), which guarantees that the generator $\mathbf{L}_1$ is non-vanishing ($\mathbf{L}_1 \neq \mathbf{0}$). According to the spectral properties of real skew-symmetric matrices, any such nonzero matrix must possess at least one pair of nonzero, purely imaginary eigenvalues. 

In the remainder of the proof, we rigorously establish that the evolution map is a transcendental matrix function of the input $u$, inherently guaranteeing an infinite-order nonlinear response. Specifically, we demonstrate that along a specific ray in the complex plane, the norm of its analytic continuation scales exponentially as $|u| \to \infty$. 
To further simplify the notation, let us define the constant real matrices $\mathbf{X} = \Delta t \mathbf{L}_0$ and $\mathbf{Y} = \Delta t \mathbf{L}_1$. The discrete update operator is given by the matrix exponential:
\begin{equation}
    F(u) = \exp(\mathbf{X} + u\mathbf{Y}).
\end{equation}

Let $D = 4^N$ denote the dimension of the Pauli-Liouville matrices. Since $\mathbf{Y} \neq \mathbf{0}$ is a strictly real skew-symmetric matrix, it must possess at least one nonzero, purely imaginary eigenvalue. Let $\beta \neq 0$ be one such eigenvalue of $\mathbf{Y}$. Consider the characteristic equation of the generator $\mathbf{X} + u\mathbf{Y}$:
\begin{equation}
    \det[\lambda I - (\mathbf{X} + u\mathbf{Y})] = 0.
\end{equation}
Introducing the asymptotic variables $t = 1/u$ and $\mu = \lambda/u$, we rewrite the generator as $\mathbf{X} + u\mathbf{Y} = u(\mathbf{Y} + t\mathbf{X})$. The characteristic equation transforms into: 
\begin{equation}
    \det[\lambda I - (\mathbf{X} + u\mathbf{Y})] = u^D \det[\mu I - (\mathbf{Y} + t\mathbf{X})] = 0.
\end{equation}
Therefore, every eigenvalue branch $\lambda(u)$ of $\mathbf{X} + u\mathbf{Y}$ corresponds identically to an eigenvalue branch $\mu(t)$ of the perturbed matrix $\mathbf{Y} + t\mathbf{X}$, related by:
\begin{equation}
    \lambda(u) = u\,\mu(1/u).
\end{equation}
In the high-drive limit ($|u| \to \infty$), we treat the intrinsic reservoir dynamics $\mathbf{X}$ as a small perturbation to the dominant drive operator $u\mathbf{Y}$. By introducing the asymptotic parameter $t = 1/u$, we analyze the eigenvalues $\mu(t)$ of the matrix $\mathbf{Y} + t\mathbf{X}$. Let $\sigma(\mathbf{Y})$ denote the spectrum of $\mathbf{Y}$, defined as the set of all its eigenvalues. According to standard finite-dimensional spectral perturbation theory, for any unperturbed eigenvalue $\beta \in \sigma(\mathbf{Y})$, there exists a corresponding perturbed eigenvalue branch $\mu(t)$ that continuously tracks $\beta$ as $t$ approaches zero:
\begin{equation}
    \mu(t) = \beta + O(t) \qquad (t \to 0).
\end{equation}
Substituting $t = 1/u$, we obtain the asymptotic behavior of the corresponding eigenvalue branch $\lambda(u)$ for large $u$:
\begin{equation}
    \lambda(u) = u\beta + O(1) \qquad (|u| \to \infty).
\end{equation}
By the spectral mapping theorem, since $\lambda(u) \in \sigma(\mathbf{X} + u\mathbf{Y})$, its exponential maps directly to the spectrum of the evolution operator:
\begin{equation}
  e^{\lambda(u)} \in \sigma\!\big(e^{\mathbf{X} + u\mathbf{Y}}\big) = \sigma(F(u)).
\end{equation}
For any valid matrix norm, the norm of a matrix is strictly bounded from below by its spectral radius $\rho(F(u))$. Hence,
\begin{equation}
    \norm{F(u)} \geq \rho(F(u)) \geq \big|e^{\lambda(u)}\big| = e^{\Re \lambda(u)}. \label{app:exp_bound}
\end{equation}

Recall that $\mathbf{Y}$ is a real skew-symmetric matrix, implying that its nonzero eigenvalue $\beta$ must be purely imaginary. We can therefore write $\beta = i\omega$ for some real nonzero constant $\omega \in \mathbb{R} \setminus \{0\}$. To exploit this, we analyze the analytic continuation of $u$ along a specific ray in the complex plane, parameterized by a real radius $\rho > 0$:
\begin{equation}
    u(\rho) = -i \rho \operatorname{sgn}(\omega). \label{app:u_rho}
\end{equation}
Along this chosen ray, the magnitude is $|u| = \rho$, and the leading term of the eigenvalue becomes strictly real and positive: $u\beta = \big(-i \rho \operatorname{sgn}(\omega)\big)(i\omega) = \rho |\omega|$. Consequently, the real part of $\lambda(u)$ scales linearly with $\rho$:
\begin{equation}
    \Re \lambda(u(\rho)) = \rho|\omega| + O(1).
\end{equation}

Since $|\omega| > 0$, there exists a strictly positive constant $c \in (0, |\omega|)$ such that for all sufficiently large $\rho$, we have $\Re \lambda(u(\rho)) \geq c\rho$. Substituting this into Eq.~\eqref{app:exp_bound} yields:
\begin{equation}
    \norm{F(u(\rho))} \geq e^{c\rho}. \label{app:exp_bound2}
\end{equation}
Therefore, the norm of the matrix $F(u)$ scales exponentially as $|u| \to \infty$ along the ray parameterized by $\rho$. We now rigorously prove that this exponential growth dictates that $F(u)$ must be a transcendental matrix function. 
By definition, a matrix-valued function $F(u)$ is algebraic if and only if every matrix element $[F(u)]_{ij}$ is an algebraic function of $u$; otherwise, it is classified as transcendental. We proceed by contradiction. Assume that $F(u)$ is an algebraic matrix function. Consequently, each element $[F(u)]_{ij}$ must satisfy a non-trivial polynomial equation $P_{ij}(u, [F(u)]_{ij}) = 0$.

According to Puiseux's theorem~\cite{brieskorn1986plane}, any branch of an algebraic function can be expanded as a convergent fractional power series (a Puiseux series) in the neighborhood of infinity. A direct asymptotic consequence of this theorem is that the growth of an algebraic function is strictly bounded by a polynomial at infinity. That is, for each matrix element $[F(u)]_{ij}$, there exist positive constants $M_{ij}, R_{ij} > 0$ and a rational number $q_{ij} \in \mathbb{Q}$ such that for all $|u| > R_{ij}$:
\begin{equation}
    \left| [F(u)]_{ij} \right| \leq M_{ij} |u|^{q_{ij}}.
\end{equation}
Since all matrix norms on a finite-dimensional vector space are equivalent, the induced matrix norm $\norm{\cdot}$ is bounded by the maximum element magnitude up to a constant factor. Thus, there exists a constant $C > 0$ such that $\norm{F(u)} \leq C \max_{i,j} |[F(u)]_{ij}|$. Let $R = \max_{i,j} R_{ij}$, $M = C \max_{i,j} M_{ij}$, and $q = \max_{i,j} q_{ij}$. It follows that the overall matrix norm is strictly bounded by a polynomial for all $|u| > R$:
\begin{equation}
    \norm{F(u)} \leq M |u|^q. \label{eq:algebraic_bound}
\end{equation}

Restricting our analysis to the specific ray parameterizing the exponential growth in Eq.~\eqref{app:u_rho}, this algebraic upper bound requires that for sufficiently large $\rho$:
\begin{equation}
    \norm{F(u(\rho))} \leq M \rho^q.
\end{equation}
However, we have already established the strict exponential lower bound in Eq.~\eqref{app:exp_bound2}:
\begin{equation}
    e^{c\rho} \leq \norm{F(u(\rho))}.
\end{equation}
Combining the algebraic upper bound with the established lower bound, we obtain:
\begin{equation}
    e^{c\rho} \leq \norm{F(u(\rho))} \leq M \rho^q.
\end{equation}
Dividing the inequality by $e^{c\rho}$ (which is strictly positive) yields: $1 \leq M  \rho^q/e^{c\rho}$. 
Since $c > 0$ is a strictly positive constant, taking the limit as $\rho \to \infty$ gives: $\lim_{\rho \to \infty} M  \rho^q/e^{c\rho} = 0$. 
This implies $1 \leq 0$, which is mathematically absurd. 
This contradiction rigorously invalidates our initial assumption. The asymptotic growth of $\norm{F(u)}$ strictly violates the algebraic bounds imposed by Puiseux's theorem. Therefore, $F(u)$ cannot be an algebraic matrix function; it must contain at least one transcendental element, rendering $F(u)$ an intrinsically transcendental matrix function. 

In conclusion, the mathematical impossibility of satisfying the algebraic bound formally confirms that the state update map of the Hamiltonian-encoded reservoir assumes a fundamentally transcendental function form. By continuously embedding the temporal input directly into the non-commuting generator of the quantum dynamics, the system organically synthesizes an unbounded hierarchy of nonlinear responses. This ensures that its power series expansion cannot truncate at any finite order. Consequently, Hamiltonian encoding inherently equips the quantum reservoir with an infinite-order nonlinear expressive capacity, fundamentally overcoming the restrictive polynomial bottlenecks that cripple traditional qubit-resetting architectures.

\section{Hamiltonian-encoding QRC prediction problem}
\label{app:BQP}
In this appendix, we rigorously define the Hamiltonian-encoding QRC prediction problem and prove its BQP-hardness. 

We consider a family of quantum reservoir prediction problems parameterized
by the reservoir size $N$. An instance consists of the classical
descriptions of the following objects:
\begin{equation}
    \mathcal{I}
    =
    \left(
        N,p,M,\Delta t,
        \rho_{0},
        H,
        \{L_{\mu},\gamma_{\mu}\}_{\mu=1}^{r},
        \mathcal{O},
        \mathcal{D},
        \bm{u}^{\star},
        \lambda
    \right).
\end{equation}
Here, $N$ is the number of reservoir qubits, $p$ is the length of each finite-length input sequence, $M$ is the number of training samples, and
$\Delta t>0$ is the duration of each input interval. The initial reservoir
state $\rho_{0}$ is an $N$-qubit state specified by an efficiently
preparable classical description.

Each input sample is a real-valued sequence
\begin{equation}
    \bm{u}
    =
    (u_{1},\ldots,u_{p})
    \in\mathbb{R}^{p}.
\end{equation}
During the $k$th input interval, the reservoir evolves under the
input-dependent Lindblad generator
\begin{equation}
    \mathcal{L}(u_{k})[\rho]
    =
    -i[H(u_{k}),\rho]
    +
    \sum_{\mu=1}^{r}
    \gamma_{\mu}
    \mathcal{D}[L_{\mu}](\rho),
\end{equation}
where
\begin{equation}
    \mathcal{D}[L](\rho)
    :=
    L\rho L^{\dagger}
    -
    \frac{1}{2}
    \left\{
        L^{\dagger}L,\rho
    \right\}.
\end{equation}
The Hamiltonian encoding map
\begin{equation}
    u\longmapsto H(u)
\end{equation}
is part of the problem instance. 
The Lindblad operators $L_{\mu}$ and rates $\gamma_{\mu}\geq0$ are also
part of the instance unless explicitly fixed by the problem family.

For an input sequence $\bm{u}$, the reservoir state after all $p$ input
intervals is
\begin{equation}
    \rho_{\bm{u}}
    =
    \mathcal{E}_{u_{p}}
    \circ\cdots\circ
    \mathcal{E}_{u_{1}}(\rho_{0}),
    \qquad
    \mathcal{E}_{u}
    :=
    e^{\Delta t\mathcal{L}(u)}.
\end{equation}
Let
\begin{equation}
    \mathcal{O}
    =
    \{O_{1},\ldots,O_{d}\}
\end{equation}
be a prescribed collection of bounded local observables, with
$d=\operatorname{poly}(N)$. The corresponding reservoir feature map is
defined by
\begin{equation}
    \bm{\phi}(\bm{u})
    :=
    \left(
        \operatorname{Tr}[O_{1}\rho_{\bm{u}}],
        \ldots,
        \operatorname{Tr}[O_{d}\rho_{\bm{u}}]
    \right)^{\mathsf T}.
\end{equation}
A concrete choice used throughout this work is
\begin{equation}
    d=N,
    \qquad
    O_{i}=Z_{i},
    \qquad
    i\in[N].
\end{equation}

The labeled training set is
\begin{equation}
    \mathcal{D}
    =
    \left\{
        \left(
            \bm{u}^{(m)},y^{(m)}
        \right)
    \right\}_{m=1}^{M},
\end{equation}
where
$\bm{u}^{(m)}\in\mathbb{R}^{p}$ and
$y^{(m)}\in\mathbb{R}$. The feature matrix
$F\in\mathbb{R}^{M\times d}$ and label vector
$\bm{y}\in\mathbb{R}^{M}$ are defined by
\begin{equation}
    F_{mi}
    :=
    \phi_{i}(\bm{u}^{(m)}),
    \qquad
    \bm{y}
    :=
    \left(
        y^{(1)},\ldots,y^{(M)}
    \right)^{\mathsf T}.
\end{equation}
For a regularization parameter $\lambda>0$, the ridge-regression weight
vector is
\begin{equation}
    \widehat{\bm{w}}
    :=
    \operatorname*{arg\,min}_{\bm{w}\in\mathbb{R}^{d}}
    \left\{
        \|F\bm{w}-\bm{y}\|_{2}^{2}
        +
        \lambda\|\bm{w}\|_{2}^{2}
    \right\}.
\end{equation}
Since $\lambda>0$, the minimizer is unique and is given by
\begin{equation}
    \widehat{\bm{w}}
    =
    \left(
        F^{\mathsf T}F+\lambda I_{d}
    \right)^{-1}
    F^{\mathsf T}\bm{y}.
\end{equation}
For a specified target input
$\bm{u}^{\star}\in\mathbb{R}^{p}$, the predicted value is
\begin{equation}
    f(\mathcal{I})
    :=
    \widehat{y}_{\star}
    =
    \bm{\phi}(\bm{u}^{\star})^{\mathsf T}
    \left(
        F^{\mathsf T}F+\lambda I_{d}
    \right)^{-1}
    F^{\mathsf T}\bm{y}.
\end{equation}

All continuously valued quantities appearing in the instance, including
the entries of the input sequences, labels, Hamiltonians, Lindblad
operators, rates, the interval length, and the regularization parameter,
are specified by finite classical descriptions with polynomial bit
complexity. We restrict to admissible instances for which
\begin{equation}
    N,p,M,d,r
    =
    \operatorname{poly}(|\mathcal{I}|),
\end{equation}
the Hamiltonians and Lindblad operators are sums of polynomially many
constant-local terms with polynomially bounded norms, and
\begin{equation}
    \lambda
    \geq
    \frac{1}{\operatorname{poly}(|\mathcal{I}|)}.
\end{equation}
Here, $|\mathcal{I}|$ denotes the total length of the classical
description of the instance.

The associated approximation problem asks for an estimate
$\widetilde{y}_{\star}$ satisfying
\begin{equation}
    \left|
        \widetilde{y}_{\star}
        -
        f(\mathcal{I})
    \right|
    \leq
    \epsilon,
\end{equation}
where $\epsilon^{-1}=\operatorname{poly}(|\mathcal{I}|)$.

For complexity-theoretic purposes, we consider the corresponding promise
decision problem. In addition to $\mathcal{I}$, the input contains two
rational thresholds $a>b$ satisfying
\begin{equation}
    a-b
    \geq
    \frac{1}{\operatorname{poly}(|\mathcal{I}|)}.
\end{equation}
It is promised that either
\begin{equation}
    f(\mathcal{I})\geq a
    \qquad \text{or}\qquad
    f(\mathcal{I})\leq b.
\end{equation}
The task is to output \textsc{yes} in the former case and \textsc{no} in
the latter.


Now, we prove the $\mathsf{BQP}$-hardness of the general Hamiltonian-encoding QRC prediction problem. 
To establish this, it suffices to consider its
purely unitary subclass obtained by fixing
$\rho_{0}=|0^{N}\rangle\langle 0^{N}|$ and setting
$\gamma_{\mu}=0$ for all $\mu$. We assume that the encoding map
$u\mapsto H(u)$ is part of the problem instance and may be specified, on
the finitely many input values occurring in the instance, by a
polynomial-size classical lookup table.

We give a polynomial-time many-one reduction from the standard quantum
circuit acceptance problem. An instance of this problem consists of a
polynomial-size quantum circuit
\begin{equation}
    C=G_{L}G_{L-1}\cdots G_{1}
\end{equation}
acting on $N$ qubits, where each gate belongs to a fixed finite universal
gate set and acts on at most a constant number of qubits. Equivalently,
one may consider the promise problem of deciding whether
\begin{equation}
    \mu(C)
    :=
    \langle 0^{N}|C^{\dagger}Z_{1}C|0^{N}\rangle
\end{equation}
satisfies $\mu(C)\geq\frac{1}{3}$ or $\mu(C)\leq-\frac{1}{3}$. 
This promise problem is $\mathsf{BQP}$-complete, since the acceptance probability of a quantum circuit is an affine function of the Pauli-$Z$ expectation value of its output qubit.

For every gate occurrence $G_{j}$, choose a distinct classically
representable input value $v_{j}$ and define
\begin{equation}
    H(v_{j})
    :=
    \frac{i}{\Delta t}\log G_{j},
\end{equation}
where the Hermitian logarithm is taken only on the constant-dimensional
support of $G_{j}$. Since every finite-dimensional unitary admits a
Hermitian logarithm, this definition ensures
\begin{equation}
    e^{-i\Delta t H(v_{j})}=G_{j}.
\end{equation}
Moreover, because $G_{j}$ acts on at most a constant number of qubits,
$H(v_{j})$ is also supported on at most the same constant number of
qubits. 
The target input sequence is then defined as
\begin{equation}
    \bm{u}^{\star}
    :=
    (v_{1},v_{2},\ldots,v_{L}).
\end{equation}
With the convention that the inputs are processed from left to right,
the corresponding reservoir evolution is
\begin{equation}
\begin{aligned}
    U(\bm{u}^{\star})
    &=
    e^{-i\Delta t H(v_{L})}
    \cdots
    e^{-i\Delta t H(v_{1})} \\
    &=
    G_{L}\cdots G_{1}
    =
    C.
\end{aligned}
\end{equation}
Consequently, the first component of the target feature vector is exactly
\begin{equation}
    \phi_{1}(\bm{u}^{\star})
    =
    \mu(C).
\end{equation}

It remains to construct a training set whose feature matrix forces the
ridge-regression readout to select the first feature. To this end, enlarge
the input alphabet by introducing an identity symbol $v_{\mathrm{id}}$
and, for every qubit $i\in[N]$, two calibration symbols
$v_i^{(Y)}$ and $v_i^{(X)}$ such that
\begin{equation}
    H(v_{\mathrm{id}})=0,
    \qquad
    H(v_i^{(Y)})=\frac{\pi}{4\Delta t}Y_i,
    \qquad
    H(v_i^{(X)})=\frac{\pi}{2\Delta t}X_i. \notag
\end{equation}
Accordingly,
\begin{equation}
    e^{-i\Delta t H(v_i^{(Y)})}|0\rangle_i
    =
    |+\rangle_i,
    \qquad
    e^{-i\Delta t H(v_i^{(X)})}|0\rangle_i
    =
    -i|1\rangle_i. \notag
\end{equation}

We take all input sequences to have the common length
\begin{equation}
    p=L+N.
\end{equation}
The target sequence is padded by identity symbols,
\begin{equation}
    \bm{u}^{\star}
    =
    (v_1,\ldots,v_L,
    \underbrace{v_{\mathrm{id}},\ldots,v_{\mathrm{id}}}_{N}),
\end{equation}
so that its induced evolution remains equal to the original circuit
$C$.

For each $r\in[N]$, define two training sequences
$\bm{u}^{(r,+)}$ and $\bm{u}^{(r,-)}$. Their first $L$ entries are all
$v_{\mathrm{id}}$, so none of the circuit gates is applied. Their final
$N$ entries are specified by
\begin{equation}
    u_{L+i}^{(r,+)}
    =
    \begin{cases}
        v_{\mathrm{id}}, & i=r,\\
        v_i^{(Y)}, & i\neq r,
    \end{cases}
\end{equation}
and
\begin{equation}
    u_{L+i}^{(r,-)}
    =
    \begin{cases}
        v_r^{(X)}, & i=r,\\
        v_i^{(Y)}, & i\neq r.
    \end{cases}
\end{equation}
Since these calibration unitaries act on distinct qubits, the resulting
states are
\begin{equation}
    U(\bm{u}^{(r,+)})|0^N\rangle
    =
    |0\rangle_r
    \bigotimes_{i\neq r}|+\rangle_i
\end{equation}
and
\begin{equation}
    U(\bm{u}^{(r,-)})|0^N\rangle
    =
    -i|1\rangle_r
    \bigotimes_{i\neq r}|+\rangle_i.
\end{equation}
Therefore,
\begin{equation}
    \bm{\phi}(\bm{u}^{(r,+)})
    =
    \bm{e}_r,
    \qquad
    \bm{\phi}(\bm{u}^{(r,-)})
    =
    -\bm{e}_r,
\end{equation}
where $\bm{e}_r$ denotes the $r$th standard basis vector of
$\mathbb{R}^N$.

Ordering the $2N$ training samples as
\begin{equation}
    \bm{u}^{(1,+)},\bm{u}^{(1,-)},\ldots,
    \bm{u}^{(N,+)},\bm{u}^{(N,-)},
\end{equation}
the resulting feature matrix is exactly
\begin{equation}
    F
    =
    \begin{pmatrix}
        \bm{e}_1^{\mathsf T}\\
        -\bm{e}_1^{\mathsf T}\\
        \vdots\\
        \bm{e}_N^{\mathsf T}\\
        -\bm{e}_N^{\mathsf T}
    \end{pmatrix},
\end{equation}
and hence
\begin{equation}
    F^{\mathsf T}F=2I_N.
\end{equation}
The construction uses only $2N$ training samples and $2N+1$
additional calibration symbols, and its classical description has size
polynomial in $N$.


Then, assign the corresponding labels according to
\begin{equation}
    y^{(r,+)}:=\delta_{r1},
    \qquad
    y^{(r,-)}:=-\delta_{r1}.
\end{equation}
It follows that
\begin{equation}
    F^{\mathsf T}\bm{y}
    =
    2\bm{e}_{1}.
\end{equation}
Fixing the ridge regularization parameter to $\lambda=2$ gives
\begin{equation}
\begin{aligned}
    \widehat{\bm{w}}
    &=
    \left(
        F^{\mathsf T}F+\lambda I_{N}
    \right)^{-1}
    F^{\mathsf T}\bm{y} \\
    &=
    \left(4I_{N}\right)^{-1}
    \left(2\bm{e}_{1}\right)
    =
    \frac{1}{2}\bm{e}_{1}.
\end{aligned}
\end{equation}
Therefore, the ridge prediction on the target sequence satisfies
\begin{equation}
\begin{aligned}
    \widehat{y}_{\star}
    &=
    \bm{\phi}(\bm{u}^{\star})^{\mathsf T}
    \widehat{\bm{w}} \\
    &=
    \frac{1}{2}\phi_{1}(\bm{u}^{\star})
    =
    \frac{1}{2}\mu(C).
\end{aligned}
\end{equation}
Hence,
\begin{align}
    \mu(C)\geq\frac{1}{3}
    \quad&\Longrightarrow\quad
    \widehat{y}_{\star}\geq\frac{1}{6}, \\
    \mu(C)\leq-\frac{1}{3}
    \quad&\Longrightarrow\quad
    \widehat{y}_{\star}\leq-\frac{1}{6}.
\end{align}
Thus, choosing the decision thresholds as
\begin{equation}
    a=\frac{1}{6},
    \qquad
    b=-\frac{1}{6}
\end{equation}
defines a promise-preserving reduction from quantum circuit acceptance
to the Hamiltonian-encoding QRC prediction problem.

The construction introduces only a polynomial number of input symbols,
Hamiltonian terms, training samples, and classically specified numerical
parameters, and can be carried out by a deterministic classical
algorithm in time polynomial in the description length of $C$.
Therefore,  the general Hamiltonian-encoding QRC prediction problem is $\mathsf{BQP}$-hard.

\end{document}